\documentclass[xcolor=svgnames,11pt,reqno]{amsart}

\usepackage[T1]{fontenc}
\usepackage[round]{natbib}
\usepackage{caption,subcaption}     
\usepackage{comment,dirtytalk,anyfontsize}
\usepackage{lmodern,lipsum,appendix}
\usepackage{amsfonts,dsfont,mathrsfs,nicefrac,yhmath}
\usepackage{amsmath,amssymb,amsthm,mathtools,bm,bbm}
\usepackage{enumitem}
\usepackage{anyfontsize}
\usepackage{microtype}
\usepackage{fancyhdr}
\usepackage{relsize}
\usepackage{cancel}
\usepackage{tikz,graphicx,xcolor}
\usepackage{booktabs}
\usepackage{siunitx}

\numberwithin{equation}{section}

\newtheorem{theorem}{Theorem}[section]
\newtheorem{lemma}[theorem]{Lemma}
\newtheorem{corollary}[theorem]{Corollary}
\newtheorem{proposition}[theorem]{Proposition}
\newtheorem{assumption}[theorem]{Assumption}

\theoremstyle{definition}

\newtheorem{definition}[theorem]{Definition}
\newtheorem{example}[theorem]{Example}

\newcommand{\cX}{\mathcal{X}}
\newcommand{\cA}{\mathcal{A}}

\newcommand{\cF}{\mathcal{F}}

\newcommand{\bbR}{\mathbb{R}}
\newcommand{\p}{\mathbb{P}}

\renewcommand{\(}{\left(}
\renewcommand{\)}{\right)}
\renewcommand{\tilde}{\widetilde}

\newcommand{\figref}[1]{\figurename~\ref{#1}}

\DeclareMathOperator*{\argmax}{\arg\max}

\allowdisplaybreaks

\usepackage{float}
\usepackage{subcaption}
\usepackage{multirow}
\usepackage[left=2.5cm,top=2.5cm,right=2.5cm,bottom=2.5cm,head=3cm,headsep=1cm, foot=3cm]{geometry}

\usepackage{hyperref}
\definecolor{darkblue}{rgb}{0.1,0.1,0.9}
\definecolor{darkred}{rgb}{0.9,0.1,0.1}

\hypersetup{colorlinks=true,linkcolor=blue,citecolor=blue,urlcolor=blue}

\begin{document}

\title[Betting under Common Beliefs: The Effect of Probability Weighting]{Betting under Common Beliefs:\vspace{0.3cm}\\The Effect of Probability Weighting\\\vspace{0.1cm}}


\begin{abstract}
This paper examines the impact of introducing a Rank-Dependent Utility (RDU) agent into a von Neumann-Morgenstern (vNM) pure-exchange economy with no aggregate uncertainty. In the absence of the RDU agent, the classical theory predicts that Pareto-optimal allocations are full-insurance, or no-betting, allocations. We show how the probability weighting function of the RDU agent, seen as a proxy for probabilistic risk aversion that is not captured by marginal utility of wealth, can lead to Pareto optima characterized by endogenous betting, despite common baseline beliefs. Such endogenous betting at an optimum leads to uncertainty-generating trade arising purely from heterogeneity in the perception of risk, rather than in beliefs. Our results formalize the intuitive understanding that probability weighting can act as an endogenous source of belief heterogeneity, and provide a new behavioral foundation for the coexistence of common beliefs and speculative behavior, in an environment with no initial aggregate uncertainty. Interpreting the RDU agent's nonlinear weighting function as an ``internality'' prompts the question of whether a social planner should intervene. We show how a benevolent social planner can nudge the RDU agent to behave closer to a vNM agent, through costly statistical or financial education, thereby  (partially) restoring the optimality of full-insurance allocations.
\end{abstract}


\author[Patrick Bei{\ss}ner, Tim J.\ Boonen, and Mario Ghossoub]{Patrick Bei{\ss}ner\vspace{0.15cm}\\Australian National University\vspace{0.6cm}\\
Tim J.\ Boonen\vspace{0.15cm}\\University of Hong Kong\vspace{0.6cm}\\
Mario Ghossoub\vspace{0.15cm}\\University of Waterloo\vspace{0.7cm}\\
This draft: 
\today\vspace{0.2cm}}

\address{{\bf Patrick Bei{\ss}ner}:  Research School of Economics -- The Australian National University -- Canberra, ACT 2600 -- Australia\vspace{0.15cm}}
\email{\href{mailto:patrick.beissner@anu.edu.au}{patrick.beissner@anu.edu.au}\vspace{0.4cm}}

\address{{\bf Tim J.\ Boonen}:  
Department of Statistics and Actuarial Science
 -- The University of Hong Kong -- Pokfulam  -- Hong Kong\vspace{0.15cm}}
\email{\href{mailto:tjboonen@hku.hk}{tjboonen@hku.hk}\vspace{0.4cm}}

\address{{\bf Mario Ghossoub}: University of Waterloo -- Department of Statistics and Actuarial Science -- 200 University Ave.\ W.\ -- Waterloo, ON, N2L 3G1 -- Canada\vspace{0.15cm}}
\email{\href{mailto:mario.ghossoub@uwaterloo.ca}{mario.ghossoub@uwaterloo.ca}}

\thanks{\textit{Key Words and Phrases:} Risk Sharing, Betting, Pareto Optimality, Full-Insurance Allocations, Rank-Dependent Utility, Probability Weighting.\vspace{0.15cm} }

\thanks{\textit{JEL Classification:} C02, D86, G22. \vspace{0.15cm} }

\thanks{\textit{2010 Mathematics Subject Classification:} 91B30. \vspace{0.15cm}}

\thanks{Mario Ghossoub acknowledges financial support from the Natural Sciences and Engineering Research Council of Canada (NSERC Grant No.\ 2018-03961 and 2024-03744).}

\maketitle


\section{Introduction}

A fundamental question in economic theory is when and why economic agents engage in speculative trade (or betting). Starting from an economic environment with no aggregate uncertainty, the classical prediction under Expected-Utility (EU) theory is that risk-averse agents will refrain from any uncertainty-generating trade unless they hold heterogeneous probabilistic beliefs. Indeed, as shown by \cite{Billotetal2000,Billotetal2002} and \cite{Rigottietal2008}, when all agents are EU maximizers sharing the same belief, every Pareto-efficient allocation must be a \emph{no-betting allocation} (sometimes called a \emph{full-insurance allocation}), that is, an allocation where each individual's consumption is deterministic and equal across states. Conversely, any form of betting can be Pareto improving only if agents disagree about probability assessments. In the EU framework, therefore, heterogeneity in beliefs is both a necessary and sufficient condition for betting to be Pareto improving when starting from a full-insurance allocation.

\medskip

The prediction of the classical literature is, however, empirically unconvincing. As emphasized by \cite{Billotetal2000,Billotetal2002}, real-world economies exhibit far less betting than the classical theory predicts. If agents' beliefs differ because their subjective probabilities arise from differences in preferences, as in the subjective expected utility (SEU) model of \cite{deFinetti1937} and \cite{Savage}, then divergent preferences would generically induce heterogeneous beliefs, thereby implying speculative trade. The observed paucity of betting in the real world would therefore suggest an implausibly high degree of homogeneity in preferences. Following \cite{Billotetal2000,Billotetal2002}, we interpret this discrepancy as evidence that the classical SEU model may be too restrictive a description of actual behavior.

\medskip

\cite{Billotetal2000,Billotetal2002} argue that the failure of the SEU model is due to Ellsberg-type behavior, and they re-examine the validity of the classical model's predictions under either a multiple-priors maxmin model of decision making \`a la \cite{GilboaSchmeidler1989maxmin} (as in \cite{Billotetal2000}) or a Choquet-type model of decision making \`a la \cite{schmeidler89} (as in \cite{Billotetal2002}). \cite{BBG24}, on the other hand, argue that this might be due to Allais-type behavior, as encapsulated in the Rank-Dependent Utility (RDU) model of \cite{quiggin82,Quiggin1991}, which is based on distorting objectively given probabilities through a probability weighting function. \cite{BBG24} demonstrate how betting can arise in a market with two RDU agents having different distortions of a common baseline objective probability measure on the space: sufficiently risk-seeking probability weighting may induce betting at Pareto-efficient allocations, even without differences in baseline beliefs.

\medskip

In this paper, we take a similar stance to that of \cite{BBG24}, by starting from the initial position that the empirical failure of the classical model in predicting betting behavior in real-world markets could very well stem from Allais-type behavior, rather than Ellsberg-type behavior. We aim to show how the introduction of a single RDU agent, into an otherwise classical setting with EU agents that have common beliefs, can lead to drastically different predictions; namely, that betting might be Pareto improving, despite the presence of a common baseline probability measure on the state space. 

\medskip
 
Specifically, we study an economy with $n > 1$ agents and a deterministic aggregate endowment, in which $n-1$ agents are EU maximizers and the $n^{th}$ agent is an RDU maximizer. All agents share the same underlying probability measure on the state space, but the RDU agent distorts this measure through a nonlinear probability weighting function seen as a proxy for probabilistic risk aversion, that is, an element of risk aversion that is not captured by the marginal utility of wealth. We provide a general characterization of Pareto-efficient allocations in this mixed EU-RDU economy. We then identify conditions under which no-betting allocations remain Pareto efficient, and we show how deviations from linear probability weighting translate into well-defined regions of Pareto-efficient betting. Thereby, we demonstrate how the RDU agent's probabilistic risk aversion fundamentally alters the structure of Pareto optima. When the RDU agent's weighting function is linear, we recover the classical no-betting result of \cite{Billotetal2000,Billotetal2002}. When the weighting function departs from linearity, however, Pareto-efficient allocations can exhibit endogenous betting, that is, uncertainty-generating trade arising purely from heterogeneity in the perception of risk, rather than in beliefs. In essence, our results formalize the intuitive understanding that probability weighting can act as an endogenous source of belief heterogeneity. Even when all agents agree on the underlying probabilities, distortions generate effective disagreements that restore the possibility of beneficial speculative trade. Our framework therefore provides a new behavioral foundation for the coexistence of common beliefs and speculative behavior in an environment with no initial aggregate uncertainty.

\medskip

We develop comparative statics for endogenous uncertainty when the RDU agent is endowed with inverse S-shaped or S-shaped probability weighting functions. Inverse S-shaped (S-shaped) weighting functions arise when agents overweight (underweight) extreme gains and losses, and are studied in \cite{tversky1992advances}, \cite{Prelec98}, and, more recently, \cite{bleichrodt2023testing}. We show that a shift from linear to inverse S-shaped or S-shaped probability weighting functions systematically splits efficient allocations into a risk-free component and a risky component. The analysis provides a simple principle for identifying when full insurance prevails and when betting emerges, and it links the extent of probability distortions to the shape of the distribution functions of optimal payoffs. To make these mechanisms transparent, we offer a tractable characterization under exponential utility and \cite{Prelec98} type weighting specifications. Moreover, we evaluate welfare via certainty equivalents and discuss how compensating transfers based on the Kaldor-Hicks criterion (e.g., \cite{kaldor1939}) can convert a specific efficient allocation into the set of all Pareto-optimal allocations. 

\medskip

When agents use Prelec weighting functions, we characterize the optimal risk allocations in closed form. Interestingly, the more inverse S-shaped the weighting function becomes, the larger the probability of the full-insurance event, hence implying that the risk allocation yields a deterministic payoff with a larger probability. However, if the Prelec weighting function is S-shaped, then making it more S-shaped renders the probability of the full-insurance event smaller. For the Hurwitz weighting function used in \cite{bleichrodt2023testing}, we find that the probability of the atom increases with the ambiguity index, but the slope is less steep when the perceived ambiguity is larger.

\medskip

The aforementioned emergence of endogenous betting despite common beliefs, as a result of probability distortions, prompts the question of whether a social planner should intervene. Guided by libertarian paternalism, in the sense of  \cite{thaler2009nudge} and \cite{bernheim2018behavioral}, we consider a benevolent social planner who ``nudges'' the RDU agent through statistical or financial education. Education  directly impacts  the probability weighting function, pushing it toward linearity, so that the distorted beliefs move back toward the baseline belief. Education is costly, resulting in a reduction of the aggregate endowment in the economy. The outcome is to nudge the RDU agent to behave closer to an EU agent, thereby  (partially) restoring the risk-free full-insurance (no-betting) allocation of the reduced constant aggregate endowment. 

\medskip

Our conceptual framework for nudging the RDU agent builds upon the distinction between experienced utility and decision utility, which has received renewed attention since \cite{kahneman1997back}. In the same vein, \cite{chetty2015behavioral} highlights how behavioral economics generates new welfare implications. Indeed, behavioral biases often lead to differences between a policymaker's perspective, typically rooted in an agent's experienced utility, and the agent's decision utility. Therefore, an expansion of the set of policy tools that accounts for such differences is needed. 
Our view is that the RDU agent's nonlinear weighting function is an ``internality'' (e.g., \cite{Herrnsteinetal93}), that is, a within-person externality stemming from the individual's failure to successfully pursue their own interests.\footnote{This interpretation relates our work to  a broader class of behavioral biases  that address welfare consequences, such as low saving rates (\cite{thaler2004save}), changes in reference points (\cite{reck2023welfare}), and deficits in financial literacy (\cite{hastings2013financial,bernheim2018behavioral}).} The probability weighting function becomes a measurable proxy for the degree of misalignment between experienced and decision utility, offering both diagnostic insight and a lever for policy intervention.

\medskip

The remainder of this paper is organized as follows. Section~~\ref{SecSetting} and Section~~\ref{SecBettingPO} develop the model and characterize Pareto optima. Section~~\ref{Section:noisePO} examines the endogenous uncertainty in detail, providing comparative-static results linking distortions in probability weighting to the variance and skewness of efficient allocations. Section~~\ref{Sec:nud} introduces the planner's intervention and evaluates its effectiveness in restoring full insurance. Section~~\ref{sec:6} concludes. All proofs appear in the \hyperlink{LinkToAppendix}{Appendix}.

\bigskip
\section{The Economy}
\label{SecSetting}

\subsection{Preferences}
Let $(\Omega, \cF,\p)$ be a nonatomic probability space and let $\cX \subset L^\infty(\Omega, \cF,\p)$ be a given collection of random variables containing  all payoffs available to the economic agents. Consider a pure-exchange economy under uncertainty, with $n$ economic agents and no aggregate uncertainty. The first agent is a Rank-Dependent Utility (RDU) maximizer (as in \cite{quiggin82,Quiggin1991}), with a preference representation given by the functional $U_1: \cX\to \mathbb{R}$ defined by
$$
U_1(X)=\int u_1(X) \, \textnormal{d}\,T\circ\p, \ \ \forall \, X \in \cX,
$$

\noindent for some utility function $u_1$ and a probability weighting function $T: [0,1] \to [0,1]$, where integration is in the sense of Choquet. 

\medskip

The agents $i = 2, \ldots, n$ are Expected-Utility (EU) maximizers, each with a utility function $u_i$ and preferences represented by the functional $U_i: \cX\to \mathbb{R}$, with 
$$U_i(X)=\int u_i(X) \, \textnormal{d}\p.$$

\noindent Hence, all agents have the same underlying beliefs  $\p$ on $(\Omega, \cF)$. We make the following standard assumption.
  
\medskip

\begin{assumption}
\label{AssumpUtDist}
The function $T$ is  twice differentiable, strictly increasing, and such that $T(0)=0$ and $T(1)=1$. Moreover, each utility function $u_i$ is  twice differentiable, strictly concave, and  increasing. 
\end{assumption}

\medskip

\noindent Note that $U_1$ fails to be concave in general, unless $T$ is convex. Additionally, concavity of $u_1$ and convexity of $T$ are tantamount to strong risk aversion of the RDU agent (e.g., \cite{Chewetal1987}).
 
\medskip

\subsection{Allocations}
The initial endowments of the agents are denoted by $\zeta_1,\ldots,\zeta_n \in \cX$, assumed nonzero. Additionally,  we assume that there is no aggregate uncertainty in the economy, that is, the aggregate  endowment is constant: 
$$\sum_{i=1}^n\zeta_i=\mathtt{w}\in\mathbb{R}.$$

\noindent The set of feasible allocations is given by
$$\mathcal A(\mathtt{w}):=\Big\{\mathbf X:=\(X_1, \ldots, X_n\)\in \cX^n:\sum_{i=1}^nX_i=\mathtt{w}\Big\},$$

\noindent that is, the set of all vectors in $\cX^n$ that add up to the constant aggregate endowment $\mathtt{w}$.

\medskip

\begin{definition}
An allocation $\mathbf X \in \mathcal A(\mathtt{w})$ is said to be Pareto Optimal (in short, PO) if there is no allocation $\mathbf X'\in \mathcal A(\mathtt{w})$ such that  $U_i( X'_i) \geq U_i(X_i),$ for all $ i \in \{1, \ldots, n\},$ with at least one strict inequality. 

\medskip

Similarly, for each $Y\in \cX$, an allocation 
$$
\mathbf{X}_{-1} \in 
{\mathcal A}_{-1}(Y):=\Big\{\mathbf{X}_{-1}:=\(X_2, \ldots, X_n\) \in \cX^{n-1}:\sum_{i=2}^nX_i=Y\Big\}
$$ 
is said to be $Y$-Pareto Optimal (in short, $Y$-PO) if there is no allocation $\mathbf X' \in {\mathcal A}_{-1}(Y)$ such that 
$U_i(X'_i) \geq U_i(X_i)$, for all   $ i \in \{2, \ldots, n\},$ with at least one strict inequality.
\end{definition}

\bigskip

\section{Betting and Efficiency}
\label{SecBettingPO}

\subsection{The Utilitarian Welfare Function}
Our focus is on the following social welfare maximization problem:
\begin{equation}
\label{eq:gen-prob}
\sup_{\mathbf{X}\in \mathcal A(\mathtt{w})}
\left[U_1(X_1)+ \sum_{i=2}^n\lambda_i \, U_i(X_i)\right],
\end{equation}

\medskip

\noindent for a given vector of weights $\lambda\in\bbR^{n-1}_{++}$. We now make the following observation (see \cite{Dana93,Dana2011}), in light of Assumption \ref{AssumpUtDist}.

\medskip

\begin{proposition}
\label{PropPOMax}
Solutions to \eqref{eq:gen-prob} are PO. Moreover, if $T$ is convex, then any PO is a solution to \eqref{eq:gen-prob}, for some vector of weights $\lambda\in\mathbb{R}_{++}^{n-1}$.
\end{proposition}

\medskip

Set $X={X}_1$ and define, for a given $\lambda\in\mathbb{R}_{++}^{n-1}$, the ``inner problem'':
\begin{align}
\label{InnerProb}
U_\lambda(\mathtt{w}-X):=\sup_{\mathbf{X}_{-1}\in{\mathcal A}_{-1}(\mathtt{w}-X)}\ U^\lambda(\mathbf{X}_{-1}),\ \  \text{where} \ \ U^\lambda(\mathbf{X}_{-1})
:=
\sum_{i=2}^n\lambda_i \, U_i(X_i).
\end{align}

\noindent Then problem \eqref{eq:gen-prob} can be reformulated as:
\begin{align}
\label{eq:gen-prob2}
\sup_{X\in \cX} \, \left[U_1(X)+U_\lambda(\mathtt{w}-X)\right].
\end{align}

\medskip

\noindent By standard results on PO allocations for risk-averse EU maximizers, solutions to \eqref{eq:gen-prob2} are characterized by the Borch rule (see \cite{borch1962} and \cite{Wilson1968}). Specifically:

\medskip

\begin{proposition} 
\label{PropInner1} 
The following holds, for any $\mathbf{X}_{-1} \in 
{\mathcal A}_{-1}(Y)$:

\medskip

\begin{enumerate}
\item $\mathbf{X}_{-1}$ is $(\mathtt{w}-X)$-PO if and only if there exists some $\lambda\in\mathbb{R}_{++}^{n-1}$ such that $\mathbf{X}_{-1}$ is optimal for \eqref{InnerProb}.

\medskip

\item $\mathbf{X}_{-1}$ is optimal for \eqref{InnerProb} if and only if 
$\lambda_i \, u_i'(X_i)=\lambda_j \, u_j'(X_j),$    for all $i,j \in \{2, \ldots, n\}$.
\end{enumerate}
\end{proposition}

\medskip

In particular, if $\mathbf{X}_{-1}\in{\mathcal A}_{-1}(\mathtt{w}-X)$ is $(\mathtt{w}-X)$-PO then it must be comonotonic\footnote{$Y,Z \in \cX$ are said to be comonotonic if $\left[Y(\omega)-Y(\omega
')\right]\left[Z(\omega)-Z(\omega')\right]\ge0$, for all $\omega,\omega
'\in \Omega$. Moreover, a vector is comonotonic if it is pairwise comonotonic.}, by monotonicity of marginal utilities (Assumption \ref{AssumpUtDist}). Therefore, Pareto-optimal allocations are comonotonic vectors, and hence it suffices to look for solutions to \eqref{InnerProb} that are comonotonic. For a given $Y\in \cX$, let ${\mathcal A}_{-1}^C(Y)$ denote the collection of comonotonic allocations in ${\mathcal A}_{-1}(Y)$:
$$
\cA_{-1}^C(Y) 
:= \left\{
\mathbf{X}_{-1}\in \cA_{-1}(Y) : \, \mathbf{X}_{-1} \hbox{ is comonotonic} 
\right\}.
$$


\noindent Then in light of the above,
\begin{align}
\label{InnerProbComEq}
U_\lambda(\mathtt{w}-X)
=\sup_{\mathbf{X}_{-1}\in{\mathcal A}^C_{-1}(\mathtt{w}-X)}\ U^\lambda(\mathbf{X}_{-1}).
\end{align}

\medskip

\subsection{Quantile Formulation of Welfare Functions}

By nonatomicity, the probability space supports a random variable $\mathtt{U}$ that is uniformly distributed over $(0,1)$ \cite[Proposition D.16 \& Lemma D.17]{FollmerSchied2025}. The following result allows us to exploit an important property of comonotonic allocations that leads to a recasting of problem \eqref{InnerProbComEq} in terms of an optimization of social welfare functions over quantiles, rather than over random variable. This gives an analytically tractable formulation that will yield a closed-form characterization of Pareto optima.  

\medskip

For any $X\in \cX$, let $F_X$ and $F_X^{-1}$ denote the cumulative distribution function and (left-continuous) quantile function of $X$, respectively, defined by:
$$F_X(x) := \p\(s\in S: X(s) \leq x\)
\ \ \hbox{ and } \ \ 
F^{-1}_{X}\(t\) := \inf\Big\{x \in \mathbb{R} \ \Big | \ F_{X}\(x\) \geq t\Big\}.$$


\begin{proposition}
\label{PropComChar}
For a given $X \in \cX$, the following are equivalent: 

\medskip

\begin{enumerate}
\item[(1)] $\mathbf{X}_{-1}\in{\mathcal A}_{-1}^C(\mathtt{w}-X)$.

\medskip

\item[(2)] There exists some $\mathtt{U}\sim Uni(0,1)$ that is comonotonic with $\mathtt{w}-X$, such that $X_i = F_{X_i}^{-1}(\mathtt{U})$, a.s., for each $i \in \{2, \ldots, n\}$.
\end{enumerate}
\end{proposition}

\medskip

\noindent Therefore, for a given $\mathbf{X}_{-1}\in{\mathcal A}_{-1}^C(\mathtt{w}-X)$, and for each $i \in \{2, \ldots, n\}$, 
\begin{align*}
U_i(X_i)
=\int_0^1 u_i\left(F_{X_i}^{-1}(t)\right) \, \textnormal{d}t.
\end{align*}


\noindent Consequently,
\begin{align*}
U^\lambda(\mathbf{X}_{-1})
=  \int_0^1 \sum_{i=2}^n \lambda_i \, u_i(F_{X_i}^{-1}(t)) \,  \textnormal{d}t,
\end{align*}

\noindent and for each $t \in (0,1)$, since quantiles are additive over comonotonic random variables \cite[Lemma 4.96]{FollmerSchied2025}, we obtain
$$
\sum_{i=2}^n F_{X_i}^{-1}(t) 
=  F_{\mathtt{w}- {X}}^{-1}(t).
$$


The quantile reformulation  motivates the following problem:
\begin{align}
\label{InnerProbCom}
\sup_{f_2,\ldots, f_n\in \mathcal{Q}} \, \int_0^1 \sum_{i=2}^n \lambda_i \, u_i\left(f_i(t)\right) \,  \textnormal{d}t,  \quad \text{ such that  } \ \sum_{i=2}^n f_i =  F_{\mathtt{w} - {X}}^{-1},
\end{align}

\medskip

\noindent where $\mathcal{Q}$ is the collection of all quantile functions, i.e.,
\begin{equation*}
\mathcal{Q} := \Big\{f: (0,1) \rightarrow \mathbb{R} \ \Big | f \hbox{ is nondecreasing and left-continuous} \Big\}.
\end{equation*}


\begin{proposition}
\label{PropInnerQuant}
Fix $X \in \cX$ and let $\mathtt{U}\sim Uni(0,1)$ be such that $\mathtt{w}-{X} = F_{\mathtt{w}-{X}}^{-1}(\mathtt{U})$, a.s.\ Then $\{f_i^*\}_{i=2}^n$ is optimal for \eqref{InnerProbCom} if and only if $\{f_i^*(\mathtt{U})\}_{i=2}^n$ is optimal for \eqref{InnerProbComEq}.
\end{proposition}

\smallskip

\noindent Proposition \ref{PropInnerQuant} implies that solutions to \eqref{InnerProbCom} also yield solutions of the ``inner problem'' \eqref{InnerProb}.

\medskip
\subsection{Characterization of Efficient Allocations}
We can now  provide a closed-form characterization of efficient allocations, when one agent is an RDU-maximizer. For each $i \in \{2, \ldots, n\}$, let $I_i := \left(u_i^\prime\right)^{-1}$.  For a given vector of welfare weights $\lambda \in\mathbb{R}_{++}^{n-1}$, let $J_\lambda := I_\lambda^{-1}$, where $I_\lambda(x):= \sum_{i=2}^n I_i\left(\frac{x}{\lambda_i}\right)$, for all $x \geq 0$.

\medskip

\begin{lemma}\label{PropRepAg}
The $(n-1)$-tuple  $\{ f^\diamond_i\}_{i=2}^n$ of the quantile functions below is optimal for  Problem \eqref{InnerProbCom}:
\begin{equation}\label{eqftilde}
 f^\diamond_i := I_i\Big(\lambda_i^{-1}J_\lambda\big(F^{-1}_{\mathtt{w}-{X}}\big)\Big), \ \ \text{for all } i \in \{2, \ldots, n\}.  
\end{equation}

\medskip

\noindent Moreover, if $\mathtt{U}  \sim Uni(0,1)$ is such that $\mathtt{w}-{X} = F_{\mathtt{w}-{X}}^{-1}(\mathtt{U})$, a.s., then the following holds.

\medskip

\begin{enumerate}
\item $\lambda_i \, u_i^\prime( f^\diamond_i(\mathtt{U})) =  \lambda_j \, u_j^\prime( f^\diamond_j(\mathtt{U})) = u_\lambda^\prime(\mathtt{w}-{X})$, a.s., all $i,j \in \{2, \ldots, n\}$. 

\bigskip

\item $\left\{ f^\diamond_i(\mathtt{U})\right\}_{i=2}^n$ is optimal for \eqref{InnerProb} and hence $(\mathtt{w}-{X})$--PO, and $ f^\diamond_i(\mathtt{U}) = I_i\(\frac{J_\lambda\(\mathtt{w}-{X}\)}{\lambda_i}\)$, a.s., for each $i \in \{2, \ldots, n\}$. 

\bigskip

\item 
Additionally, 
\begin{align*}
U_\lambda(\mathtt{w}-X)
=\sup_{\mathbf{X}_{-1}\in{\mathcal A}_{-1}(\mathtt{w}-X)} \ \sum_{i=2}^n\lambda_i \, U_i(X_i)
= \sum_{i=2}^n\lambda_i \, U_i\( f^\diamond_i(\mathtt{U})\)
=  E^\p[ u_\lambda\(\mathtt{w}-{X}\) ].
\end{align*}
\end{enumerate}

\noindent Here, $u_\lambda$ is the pointwise aggregate utility given by
\begin{equation}
\label{eqAggUt-1}
u_\lambda(x) := \sum_{i=2}^n \lambda_i \, u_i\left(I_i\left(\tfrac{J_\lambda(x)}{\lambda_i}\right)\right).
\end{equation}

\end{lemma}

\medskip

This representative agent formulation helps us now to solve Problem \eqref{eq:gen-prob2}. The proof of the following result is similar to \cite{BBG24}, who examine the two-person case. For completeness, we provide a proof in the Appendix. Recall first that the convex envelope of a function $f: [0,1]\to \mathbb{R}$ is the greatest convex function $\delta: [0,1]\to \mathbb{R}$ that satisfies $\delta\(x\) \leq f\(x\)$, for each $x \in [0,1]$. The function $\delta$ is affine on the set $\left\{x \in \left[0,1\right]: \delta(x) < f(x)\right\}$ (e.g., \cite{Heetal2017}).

\medskip

\begin{lemma}
\label{LemBBGnew}
A random variable ${X}_1\in \cX$ is optimal for Problem \eqref{eq:gen-prob2} if and only if
\begin{equation}
\label{eq:X1}
{X}_1 = m^{-1}_\lambda\(\delta^\prime(\mathtt{U})\), \ \text{a.s.},
\end{equation}
where $m_\lambda\(x\):=\frac{u_\lambda^\prime\(\mathtt{w}-x\)}{u_1^\prime\(x\)}, \ \forall x \in\mathbb{R}$, and $\delta$ is the convex envelope of $\tilde T(t):=1-T(1-t)$.
\end{lemma}

\medskip

Combining lemmata \ref{PropRepAg} and \ref{LemBBGnew} leads to the following result.

\medskip

\begin{theorem}
\label{ThMain}
An allocation $\mathbf{X} \in \mathcal A(\mathtt{w})$ is optimal for Problem \eqref{eq:gen-prob} if and only if
\begin{align*}
{X}_1 =m^{-1}_\lambda( \delta^\prime(\mathtt{U})) \ \ \hbox{ and } \ \  X_i= I_i\left(\lambda_i^{-1} J_\lambda\(\mathtt{w}-{X}_1\)\right), \ \ \text{for all } \, i\in\{2,\ldots,n\}.
\end{align*}
\end{theorem}

\medskip

If $T$ is convex, then $\delta(t)=t$, for $t\in[0,1]$. We then immediately recover the special case of full insurance (no-betting). 

\medskip

\begin{corollary}
\label{corsec2}
If $T$ is convex, then $\mathbf{X}$ is PO if and only if it is a full-insurance allocation. 
\end{corollary}

\medskip

When $T$ is convex, all agents in the market are risk averse, with common baseline beliefs. In this case, PO allocations are full-insurance allocations, as one would expect. If $T$ is not concave, then $\tilde T$ is not convex, and hence  $\delta$ will not coincide with $\tilde T$. Moreover $\delta$ will be affine on the set 
$$\left\{p \in \left[0,1\right]: \delta(p) < \tilde  T(p)\right\}.$$ 

\medskip

\noindent The probability mass of this event is a gauge of the degree of full insurance within a PO allocation. The probability mass of the complement of this event quantifies the degree of endogenous uncertainty, the principal source of betting in a Pareto-optimal outcome. The following section explores in more detail the impact of introducing an RDU agent, in an otherwise classical EU-economy, on the emergence of betting and the generation of endogenous uncertainty.

\bigskip

\section{Comparative Statics of  Endogenous Uncertainty}
\label{Section:noisePO}

We wish to understand how the RDU agent disrupts the full-insurance property of efficient allocations in an EU economy with common beliefs. Theorem \ref{ThMain} suggests that the derivative of the function $\delta$ describes and quantifies how endogenous uncertainty is introduced into the economy by the RDU agent in a Pareto-efficient allocation.

\medskip

\subsection{S-Shaped and Inverse S-Shaped Probability Weighting Functions}
For illustration, we consider the broad class of inverse S-shaped and S-shaped distortion functions, with a unique inflection point $1-\bar p \in (0,1)$. A distortion function $T: [0,1] \to [0,1]$ is inverse S-shaped if it is concave on $[0, 1-\bar p]$ and convex on $[1-\bar p, 1]$. It is S-shaped if it is convex on $[0, 1-\bar p]$ and concave on $[1-\bar p, 1]$. Note that if $\tilde T$ denotes the conjugate distortion function given by $\tilde T(t) := 1 - T(1-t)$, then $\tilde T$ is inverse S-shaped (resp.\ S-shaped) with inflection point $\bar p$ if and only if $T$ is inverse S-shaped (resp.\ S-shaped) with inflection point $1- \bar p$.

\medskip

For this class of distortion functions, the following result gives an explicit characterization of the convex envelope $\delta$ of the conjugate distortion function $\tilde T$. This is illustrated in Figure \ref{figDeltaEx_}.

\medskip

\begin{proposition}\label{Prop:MG} 
 \ 
 
\begin{enumerate}
\item Suppose that $T$ is S-shaped, with an inflection point of $1-\bar p \in (0,1)$, and let 
\begin{equation}
\label{EqTanPtS}
p^* := \inf \left\{p \in \left[0,1\): \tilde T^{\prime}\(p\) \geq \frac{1-\tilde T\(p\)}{1-p}\right\}.
\end{equation}

\noindent Then $p^* \in [0, \bar p]$, and the convex envelope $\delta$ of $\tilde T$ is given by
\begin{equation}\label{eq:S}
\qquad\quad\delta(p) \:=\:
\begin{cases}
\tilde T(p) &\quad \text{if } p \leq p^*, \\[4pt]
\tilde T(p^*) + \frac{1 - \tilde T(p^*)}{1 - p^*} \cdot (p - p^*) & \quad\text{otherwise.}
\end{cases}
\end{equation}

\medskip

\item Suppose that $T$ is inverse S-shaped, with an inflection point of $1-\bar p \in (0,1)$, and let
\begin{equation}
\label{EqTanPtInvS}
p^* := \sup \left\{p \in \left[0,1\): \tilde T^{\prime}\(p\) \leq \frac{\tilde T\(p\)}{p}\right\}.
\end{equation}

\noindent Then $p^* \in [\bar p,1]$, and the convex envelope $\delta$ of $\tilde T$ is given by
\begin{equation}\label{eq:inverseS}
\delta(p) \:=\:
\begin{cases}
\frac{\tilde T(p^*)}{p^*} \cdot p & \quad\text{if } p \leq p^*, \\[6pt]
\tilde T(p) &\quad \text{otherwise.}
\end{cases}\qquad\qquad\:\:
\end{equation}
\end{enumerate}
\end{proposition}

\medskip

We  can now identify the type of betting and the event on which full insurance is prevalent. When $T$ is inverse S-shaped, $\delta^\prime$ is constant on $[0, p^*]$, and hence full insurance occurs on the event $\{\mathtt{U}\in[0, p^*]\}$; see the right part of Figure \ref{figDeltaEx_}. On the other hand, for an  S-shaped distortion function $T$, $\delta^\prime$ is constant on $[p^*,1]$. This  yields a full-insurance event $\{\mathtt{U}\in[p^*,1]\}$; see the left part of  Figure \ref{figDeltaEx_}. 

\medskip

\begin{center}
\begin{figure}[!htpb]
\includegraphics[width=12cm,height=5cm,clip,trim={0.25cm 0.25cm 0.25cm 0.5cm}]{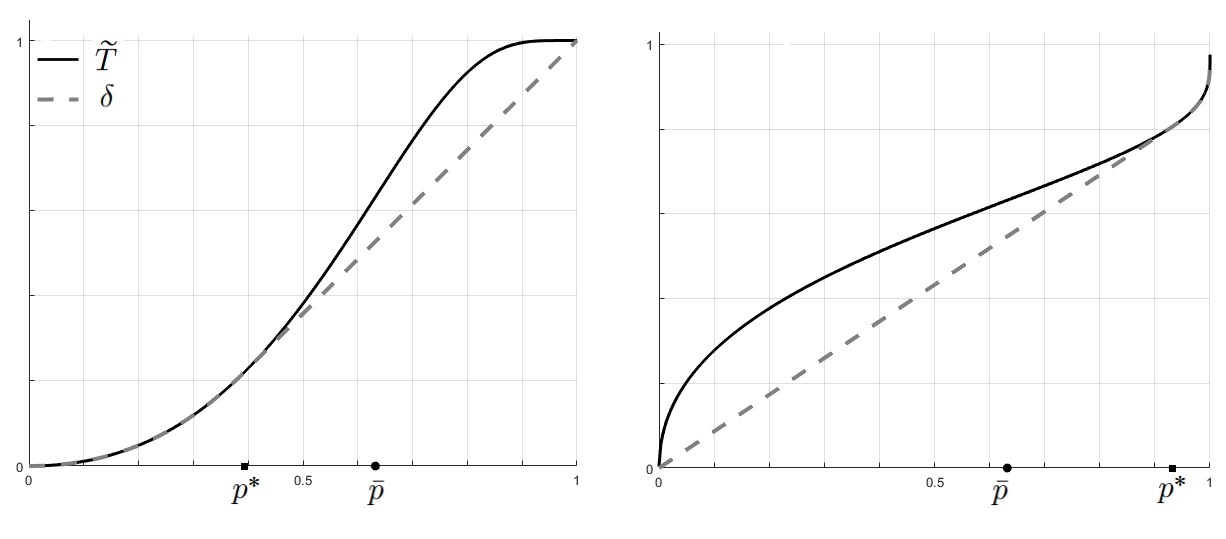}
\caption{{  Plot of an S-shaped  (left) function $\tilde T$ and an inverse S-shaped (right) function $\tilde T$, based on Prelec with $\alpha=\frac{1}{2}, 2$. In both cases the  convex envelope (dashed graph) $\delta$ with tangent point $p^*$ and inflection point $\bar p$ are highlighted.}} 
\label{figDeltaEx_}
\end{figure}
\end{center} 

\medskip
 
Additionally, by Corollary \ref{corsec2}, if $T$ is linear (or convex), then efficient allocations are full-insurance allocations. In sum, the probability mass of the full-insurance event, denoted by $\mathbb{FI}$, is 
\begin{eqnarray*}    
\mathbb{FI} = \left\{
\begin{array}{l l}
p^* & \quad \mbox{if $T$ is inverse S-shaped;}\\[4pt]
1 & \quad \mbox{if $T$ is linear (or convex);}\\[4pt]
1- p^* & \quad \mbox{if $T$ is  S-shaped.}
\end{array} \right.
\end{eqnarray*}

\medskip

\noindent In other words, $1-\mathbb{FI}$ captures the degree of endogenous uncertainty. 

\bigskip

By Theorem \ref{ThMain}, an efficient allocation $\mathbf{X}$ consists of random variables that have a mixed distribution, combining a continuous density and a discrete point mass $\mathbb{FI}$, corresponding to the full-insurance part of the allocation. The following result provides a characterization of the (defective) probability density function of the endogenous uncertainty in closed form, for the case of inverse S-shaped distortion functions.

\medskip

\begin{proposition}[Inverse S-Shaped Distortion]
\label{prop:pdf}
Suppose that the RDU agent's distortion function $T$ is inverse S-shaped, and let $p^*$ be as in \eqref{EqTanPtInvS}. Let $I(x) :=m_\lambda^{-1}( T^\prime(1-x))$,  where $m_\lambda(x)=\frac{u_\lambda^\prime(\mathtt{w}-x)}{u_1^\prime(x)}$. Let $\mathtt{U} \sim Uni(0,1)$, and let $$X_1 = m^{-1}_\lambda( \delta^\prime(\mathtt{U}))$$
be the optimal payoff of the RDU agent, where $\delta$ is the convex envelope of $\tilde T$. Then $X_1$ has a mixed distribution consisting of:

\medskip

\begin{enumerate}
\item An atom of mass $p^*$ at the point
$$x_0 := m_\lambda^{-1}\left(\frac{\tilde T(p^*)}{p^*}\right).$$

\item An absolutely continuous component supported on $I([p^*,1])$, with probability density given by
\begin{equation}
\label{fX}
f_1(x)
:=
\frac{- m_\lambda'(x)}
{T''\left((T')^{-1}(m_\lambda(x))\right)},
\ \ \forall \, x\in I([p^*,1]).
\end{equation}
\end{enumerate}
\end{proposition}

\medskip

Clearly, since $T$ is inverse S-shaped, $1-\mathbb{FI}= 1 - p^* = \int \mathbf{f}_1(x) \, dx.$ By Theorem \ref{ThMain}, we obtain the densities $\mathbf{f}_j$ also for $X_j$, $j\geq 2$. Figure  \ref{fig:1b} below provides an illustration of the risky part  and the full-insurance part.

\medskip
\subsection{The Case of a Prelec Probability Weighting Function}\label{Section:noisePO.2}
Consider the case of Prelec probability weighting functions with parameter $\alpha >0$, 
\begin{equation}
T\(p\) =\exp(-(-\ln(p))^{\alpha}), \ \ \forall p \, \in [0,1], 
\label{eq:prelec}
\end{equation} 
introduced and characterized axiomatically by \cite{Prelec98}. 
This function is inverse S-shaped when $\alpha\in(0,1)$, linear when $\alpha=1$, and S-shaped when $\alpha>1$. It then follows that $\tilde T$ from Lemma \ref{LemBBGnew} is given by
$
\tilde T(p) = 1- \exp(-(-\ln(1-p))^{\alpha}).
$  Moreover, if $T$ is (inverse-) S-shaped, then so is $\tilde T$. Hence, if $\alpha\in(0,1)$, i.e., inverse S-shaped, then $\delta$ has the form in \eqref{eq:inverseS}; and if $\alpha>1$ then $\delta$ has the form in \eqref{eq:S}; see again   \figref{figDeltaEx_}.

\medskip

Based on  the unique tangent  point $p^*=p^*(\alpha)$ where the convex envelope binds,  we can now identify the magnitude of the $\mathbb{FI}$ event depending on the parameter $\alpha$. We  have
\begin{eqnarray*}
\mathbb{FI}(\alpha) &=&
\begin{cases}
p^*(\alpha), & \quad\text{if } \alpha < 1,\\[2pt]
1,           & \quad\text{if } \alpha = 1,\\[2pt]
1 - p^*(\alpha), & \quad\text{if } \alpha > 1,
\end{cases}\qquad\qquad\qquad\qquad\qquad\qquad
\end{eqnarray*}
where $p^*(\alpha)$ can be derived from the following result.

\medskip

\begin{corollary}\label{cor:palpha} Let the  RDU agent have a Prelec weighting function $T$.
\smallskip
\begin{enumerate}
\item If $T$ is $S$-shaped, i.e., $\alpha>1$, then  
$p^*(\alpha)
= 1 - 
\exp(-{\alpha}^{-\frac{1}{\alpha-1}}).$

\medskip

\item If $T$ is inverse $S$-shaped, i.e., $\alpha<1$, then
$p^*(\alpha) = 1 - \exp(-x),$  where $x$ is the unique solution of the equation
$
x = \ln \left( \alpha x^{\alpha-1} e^x (1 - e^{-x}) + 1 \right)^{\frac{1}{\alpha}}$.
\end{enumerate}
\end{corollary}

\medskip

This result is illustrated in Figure \ref{fig:1}. We observe a counter-intuitive implication of the Prelec weighting functions. Apart from the two jumps around $\alpha=1$, we observe, for the inverse S-shaped region, that an increase in the nonlinearity leads to an increase in the probability mass of the full-insurance event.
This pattern also holds true for the inverse S-shaped function:
\begin{equation}\label{TK-w}
T\(p\) =\frac{p^\gamma}{(p^\gamma+(1-p)^\gamma)^{1/\gamma}}, \ \forall p \in [0,1], 
\end{equation} 

\noindent which is introduced by \cite{tversky1992advances}, and is inverse S shaped for $\gamma\in (0.279,1)$, as per \cite{riegerwang2006}.
In Section 4.4, we consider another class of weighting functions, introduced in    \cite{bleichrodt2023testing} with different comparative statics of $\mathbb{FI}$.

\medskip

\begin{figure}[!htpb]
        \centering
    \begin{subfigure}[b]{0.45\textwidth}
      \includegraphics[width=8cm,height=5cm,clip]{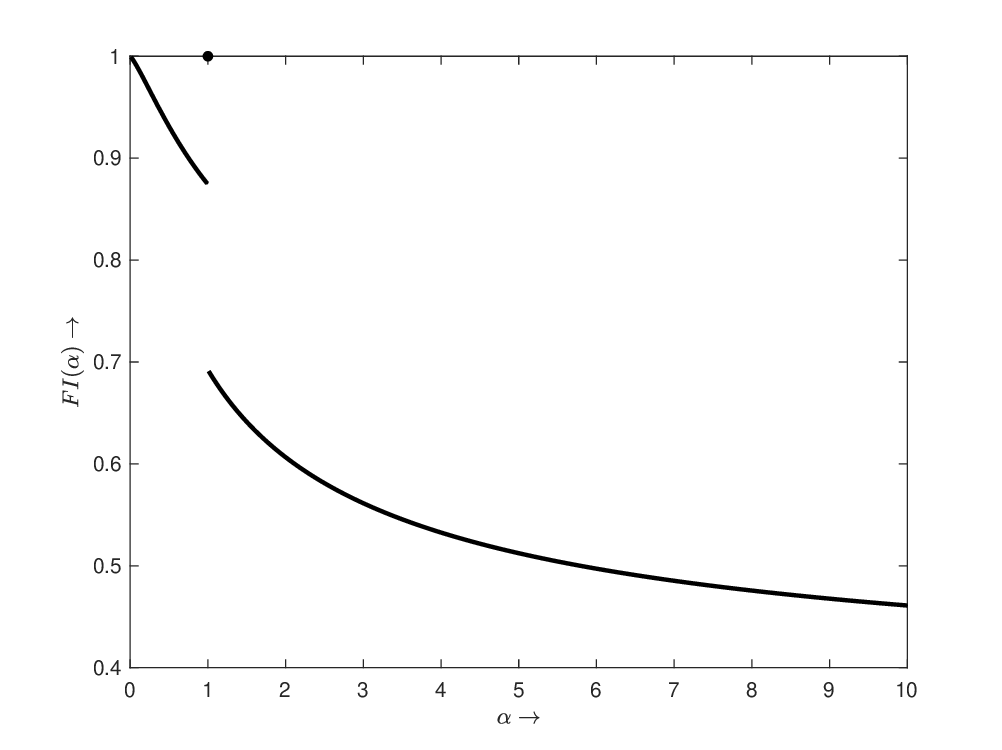}
      \end{subfigure}
    \begin{subfigure}[b]{0.45\textwidth}
       \includegraphics[width=8cm,height=5cm,clip]{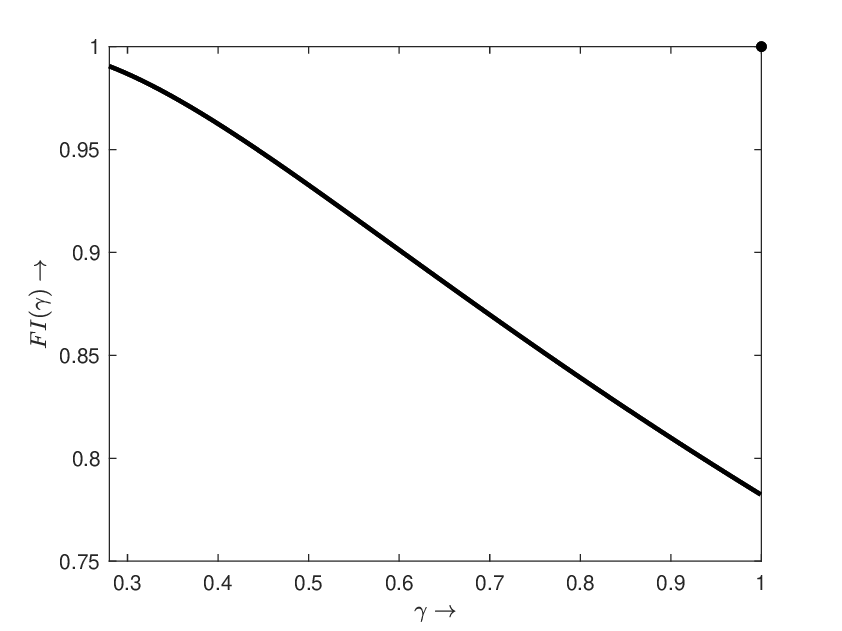}
       \end{subfigure}
\caption{\small{ Figure (a) shows the probability of the full-insurance event $\mathbb{FI(\alpha)}$ as a function of  $\alpha\in [0,10]$, for the Prelec weighting function. Figure (b) shows the probability of the full-insurance event $\mathbb{FI}=p^*$ as a function of  $gamma\in [0,10]$, for the Tversky-Kahneman weighting function, defined in \eqref{TK-w}.}
}
\label{fig:1}
\end{figure}

\medskip

\subsection{Closed-Form Solutions for a Prelec RDU Agent}
In the following, we present closed-form solutions for Pareto efficiency. Consider therefore  a more specific setup with aggregate wealth $\mathtt{w}=0$ and CARA utility for each agent $i\in\{1,2,\ldots,n\}$ given by
$$u_i(x)=-\frac{1}{\beta_i}\exp({-\beta_i x}), \text{ with } \beta_i>0.$$

\medskip

\begin{corollary}\label{Cor_explict}
Let $\overline\beta :=(\beta_2^{-1}+\cdots+\beta_n^{-1})^{-1}$, and define the rescaled risk preference parameters $\overline \beta_j := \frac{\beta_j}{\overline{\beta}}>0$ , for each $j$. Let $I(x):=\ln(\delta^\prime(x))$ and let $\ln(\sum_{i=2}^n\lambda_i)$ be a deterministic side payment. Then  
\begin{equation}
X_1
= \frac{1}{\beta_1+\overline \beta}\Big(I\(\mathtt{U}\)-\ln(\sum_{i=2}^n\lambda_i)\Big),\qquad\qquad\qquad\qquad\qquad\qquad\quad
\qquad
\label{eqYStarEx}
\end{equation}  
\begin{equation}\label{eqYStarEx2}
\text{and} \quad\qquad X_j
= \underbrace{-\frac{1}{\overline\beta_j}\cdot \frac{1}{\beta_1+\overline \beta} \ I\(\mathtt{U}\)}_{ X^\sim_j \:\text{ random }}\:+ \:\underbrace{\frac{1}{\beta_j}\Big(\ln(\lambda_j)- \sum_{k=2}^n \frac{1}{\overline\beta_k}\ln(\lambda_k)\Big)}_{X_j^\bullet\:\text{ side payment }},\quad
j=2,3,\ldots,n.
\end{equation}  
\end{corollary}

\medskip

This result delivers an explicit decomposition of optimal consumption into an atomless random variable $X^{\sim}_j$ and an atom $X^{\bullet}_j$. Moreover, the role of the vector $\lambda$ is to determine the zero-sum deterministic side payments to the agents. An increase in the aggregate risk aversion $\beta_1+\overline \beta$ moves $m_\lambda^{-1}$ from Lemma \ref{LemBBGnew} closer to the $x$-axis, and thus PO allocations are in turn closer to zero.

\medskip

\begin{figure}[!htpb]
\includegraphics[width=14cm,height=6cm,clip]
{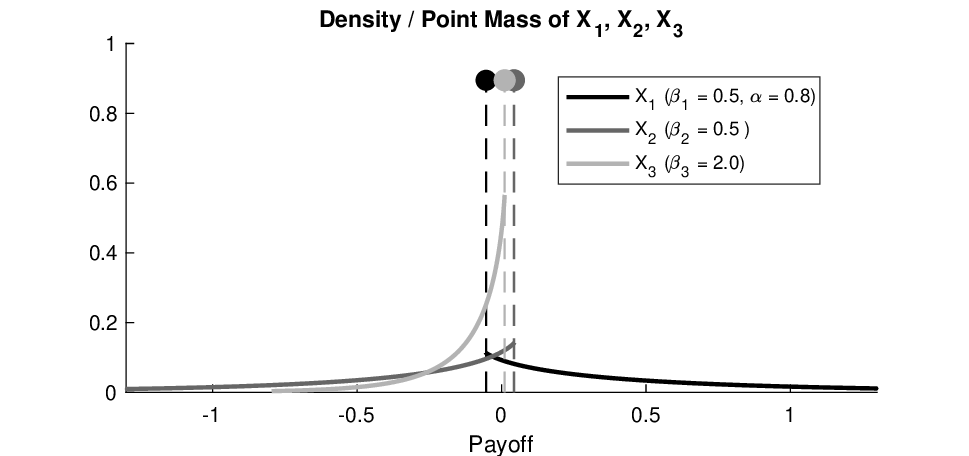}
\caption{\small{  Baseline case density function,  with $\alpha=0.8$, $\beta_1=0.5=\beta_2$,  and $\beta_3=2$.}}
\label{fig:1b}
\end{figure}

\medskip

Figure \ref{fig:1b} illustrates the density parts $$(\mathtt{f}^\alpha_1,\mathtt{f}^\alpha_2,\mathtt{f}^\alpha_3)$$ 

\noindent of the optimal allocations, when the RDU agent has an inverse S-shaped distortion function with $\alpha=0.8$.  $\lambda_2$ and $\lambda_3$ are chosen  to guarantee no side payments, i.e., 
$$X_1=\frac{1}{\beta_1+\overline \beta} \ I(\mathtt{U}), \quad \textnormal{ and }\quad  X_j=-\frac{ \overline \beta}{\beta_j}\cdot X_1, \quad j=2,3.$$ 
This implies that $\lambda_2+\lambda_3=1$, and $\frac{\beta_2-\overline\beta}{\beta_2} \ \ln(\lambda_2)=\frac{\overline\beta}{\beta_3} \ \ln(\lambda_3)$.
This is without loss, as any zero-sum side payments can be obtained by specific choices of $\lambda$. For the exponential utilities, we assume that $\beta_1=\frac{1}{2}=\beta_2$ and $\beta_3=2$. For $\lambda_2+ \lambda_3=1$, $\ln(\lambda_2+ \lambda_3)=0$ and thus the side payment to Agent 1 is equal to zero, implying that $X_1=\frac{1}{\beta_1+\overline \beta} \ I(\mathtt{U})$. Specifically, on the interval $[0,p^*]\approx[0,\frac{9}{10}]$, PO allocations are full-insurance allocations, which holds true because $\delta$ is linear on the domain $[0,p^*]$, as in the left part of Figure \ref{figDeltaEx_}. This finding holds true for any inverse S-shaped distortion function $\tilde T$ (see equation \eqref{eq:inverseS}).  Moreover, $\beta_i$ is a factor that is multiplied with the risk exposure: smaller values of $\beta_i$ lead to larger absolute values of the risk exposure. Moreover, for this RDU agent, the distribution of $X_1$ is positively (right-)skewed, due to the over-weighting of small probabilities (and absence of loss aversion). The other two EU agents show a negative skewness in their optimal payoffs $X_2$ and $X_3$, due to risk aversion and  no over-weighting of small probabilities.

\medskip
\subsection{Examples}

The following four subsections present some more  perspective on the role of the RDU agent in the structure of sharing rules: apart from the risk preference parameters $(\beta_1, \beta_2, \beta_3)$, we modify the parameter of the Prelec weighting function $\alpha$. Each example captures a three-agent economy. 
We can identify optimal consumption in closed form for the RDU agent, as in \eqref{eqYStarEx}, and also capture the underlying probability density function.

\smallskip

\subsubsection{Robustness Check}
We begin with some sensitivity analysis of the main parameter choices. Figure \ref{fig:b12}(a)  shows the case where $\beta_3$ is lowered to 0.7. Due to a lower risk aversion of Agent 3, we find that Agent 3 obtains a larger risk exposure in the tail (in absolute value), and that Agent 1 is the main counterparty of this larger risk exposure. 
Figure \ref{fig:b12}(b) shows that if $\beta_1$ is increased to 2, the risk exposures for all agents, captured by the variance, become closer to 0 due to the larger risk aversion of the RDU agent.

\medskip

\begin{figure}[!htpb]
    \centering
    \begin{subfigure}[b]{0.5\textwidth}\includegraphics[width=8.5cm,height=5cm,clip]{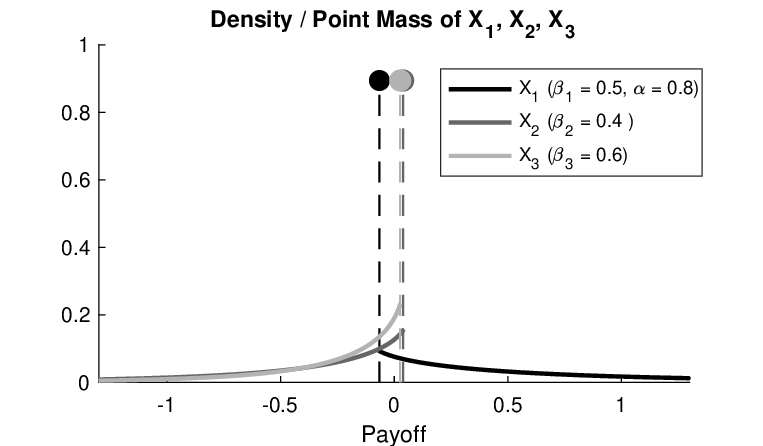}
    \end{subfigure}%
    \begin{subfigure}[b]{0.5\textwidth}
\includegraphics[width=9cm,height=5cm,clip]{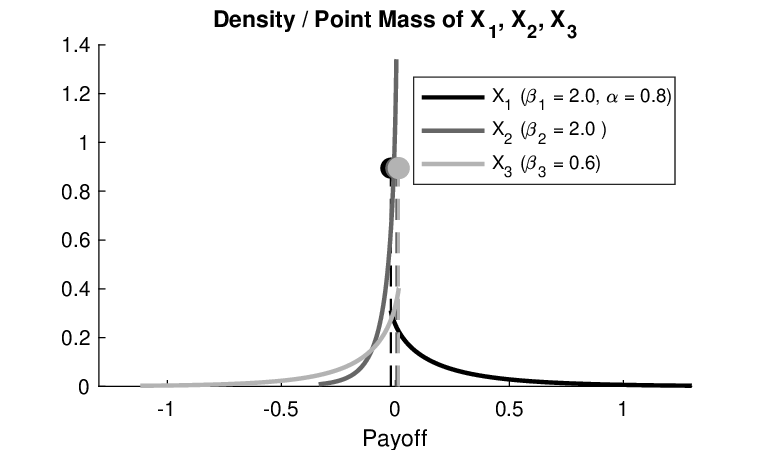}
    \end{subfigure}
    \caption{Robustness: (a) Case with $\beta_1=0.5, \beta_2=0.5$, $\beta_3=0.7$, and $\alpha=0.8$. 
    (b) Case with $\beta_1=2=\beta_3$, $\beta_2=0.5$ and  $\alpha=0.8$. 
    \medskip
    }
    \label{fig:b12}
\end{figure}

\subsubsection{Concave and Convex Case}

Next, we illustrate the risk allocations when we have a concave or convex probability weighting function. We fix $\alpha=0.5$, and re-adjust the distortion function to make it concave/convex. That is, we use the convex $\tilde T_1(p)=\frac{\tilde T( p/4)}{\tilde T(1/4)}$  and  the concave $T_2(p)=\frac{T( p/4)}{T(1/4)}$. The corresponding risk allocations are shown in Figure \ref{fig:concave}. Convex distortion functions lead to full-insurance  allocations, as per Corollary \ref{corsec2}. Moreover, concave distortion functions lead to $\delta=\tilde T_2$, and thus to strictly monotonic functions of $U$. This implies that the risk allocation does not contain an atom and there is no full-insurance event, i.e., $\mathbb{FI}=0$.

 \medskip
 
\begin{figure}[!htpb]
    \centering
    \begin{subfigure}[b]{0.5\textwidth}
        \includegraphics[width=9.5cm,height=5.cm,clip]{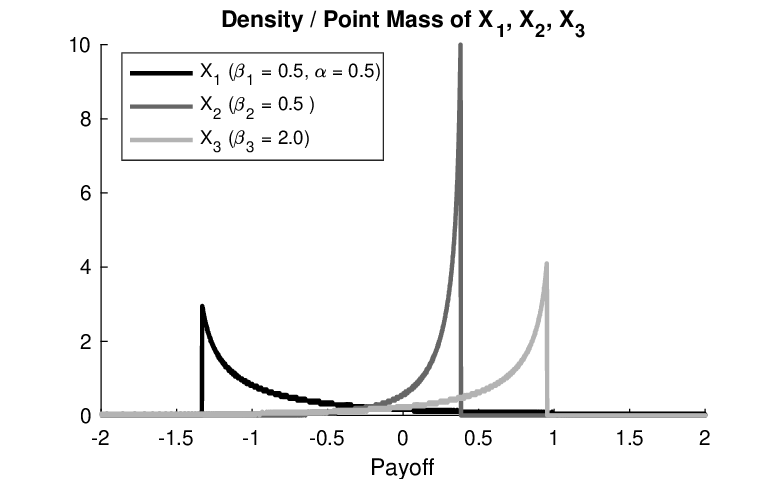}
        
    \end{subfigure}%
    \begin{subfigure}[b]{0.5\textwidth}
        \includegraphics[width=8cm,height=5.cm,clip]{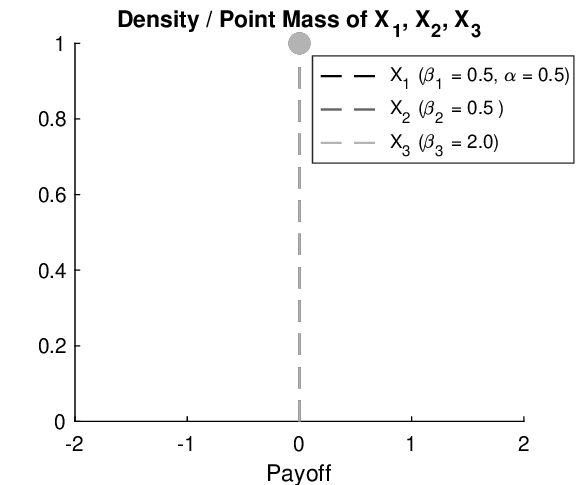}
    \end{subfigure}
    \caption{(a) Heavy betting with  concave  $T_2$. (b) Full insurance  with  convex $T_1$.}
\label{fig:concave}
\end{figure}

\smallskip
\subsubsection{From S-Shaped to Inverse S-Shaped}

Figure \ref{fig:a12} shows the optimal risk allocations when $\alpha=1.2$. Since $\alpha>1$, the corresponding probability weighting function is S-shaped. We see that there is a switch in the risk allocations: in the atomless component, the density of Agents 2 and 3, which was concentrated in the left tail when $\alpha<1$ (Figure \ref{fig:1b}), is now concentrated in the right tail when $\alpha>1$ (Figure \ref{fig:a12}). The risk allocation can be viewed as a swap of densities $\mathtt{f}^\alpha_1$ and $(\mathtt{f}^\alpha_2,\mathtt{f}^\alpha_3)$ in the atomless component, and this swap changes the role of Agent 1 from the ``buyer'' of a bet (a small chance of large gains) to a ``seller'' of a bet (a small chance of large losses). In fact, this phenomenon occurs if we change from $\alpha<1$ (inverse S-shaped) to $\alpha>1$ (S-shaped).  


\begin{figure}[!htpb]
\includegraphics[width=14cm,height=6cm,clip]{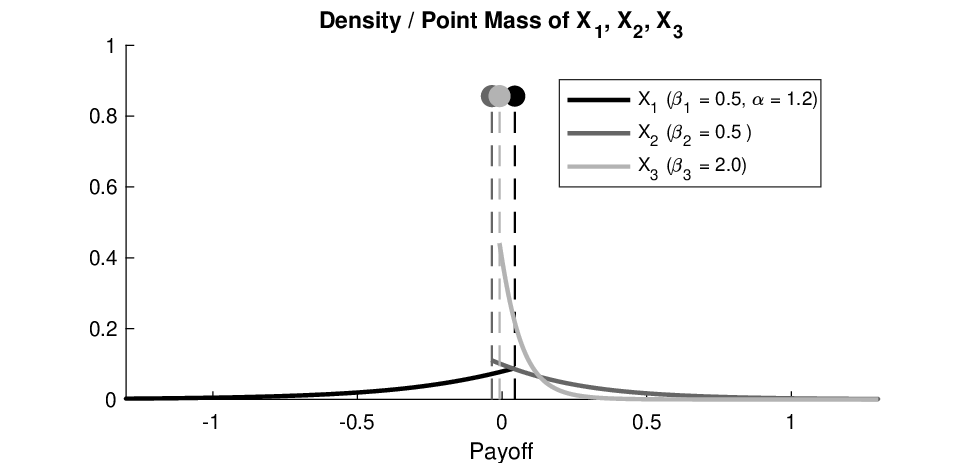}
\caption{\small{S-shaped  weighting function with $\beta_1=\beta_2=0.5$, $\beta_3=2$, $\alpha=1.2$.}}
\label{fig:a12}
\end{figure}

\medskip

This also clarifies that an S-shaped probability weighting function  creates more ``unpleasant'' left-tail uncertainty for the RDU agent, whereas she receives a positive payoff in the full-insurance event. As for the EU agents, the right-tail uncertainty is associated with a negative payoff in the full-insurance event, as shown in Figure \ref{fig:a12}. In view of Figure \ref{fig:1b} with an inverse S-shaped RDU agent, the sign of the skewness flips for each agent.

\smallskip

\subsection{Alternative Weighting Function}
Thus far, we have discussed the emergence of endogenous uncertainty  under the assumption of a Prelec-type weighting function. In this subsection, we discuss the class of weighting functions 
\begin{equation}
\label{eq:HEU}
T(p) 
= 
\gamma \ \frac{(1-\kappa)p}{(1+\kappa)(1-p)+(1-\kappa)p}
\, + \, 
(1-\gamma) \ \frac{(1+\kappa)p}{(1-\kappa)(1-p)+(1+\kappa)p},
\end{equation}
    
\noindent induced in the context of Hurwicz Expected Utility (HEU) by \cite{bleichrodt2023testing}. The parameter $\gamma\in[0,1]$ is an ambiguity index, with $\gamma=1$ corresponding to ambiguity aversion and $\gamma=0$ to ambiguity-seeking behavior. The parameter $\kappa\in[0,1]$ captures ambiguity perception, with larger values of $\kappa$ indicating greater perceived ambiguity. The probability weighting function $T$ in \eqref{eq:HEU} is a weighted average of a convex and a concave function, and it can therefore exhibit an inverse S shape.
\begin{figure}[!htpb]
\includegraphics[width=14cm,height=6cm,clip]{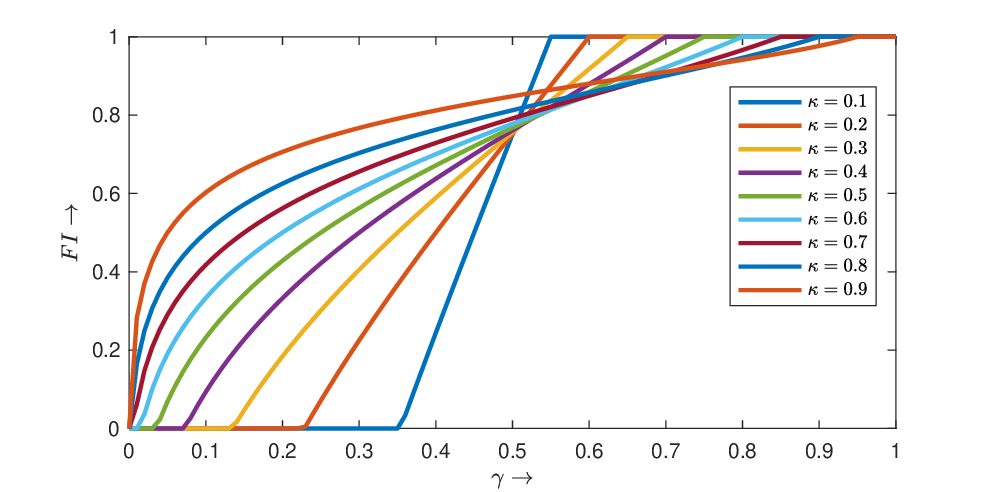}
\caption{\small{ The case of Hurwicz expected utility;  the probability of the discrete full-insurance event $\mathbb{FI}$ as a function of  $\gamma\in [0,1]$, for various values of $\kappa$. }}
\label{fig:HEU}
\end{figure}

\medskip

Figure \ref{fig:HEU}, as a counterpart of Figure \ref{fig:1} for the Prelec case, displays the probability of the discrete full-insurance event as a function
of $\gamma$, for various values of $\kappa$. We find that as $\gamma$ increases, implying more ambiguity aversion, the probability of the discrete full-insurance event increases, except when $\kappa$ is small (low perceived ambiguity). Moreover, for larger perceived ambiguity, this effect is stronger for small $\gamma$, but increases more slowly for larger $\gamma$. 

\medskip

\begin{figure}[!htpb]
\includegraphics[width=14cm,height=6cm,clip]{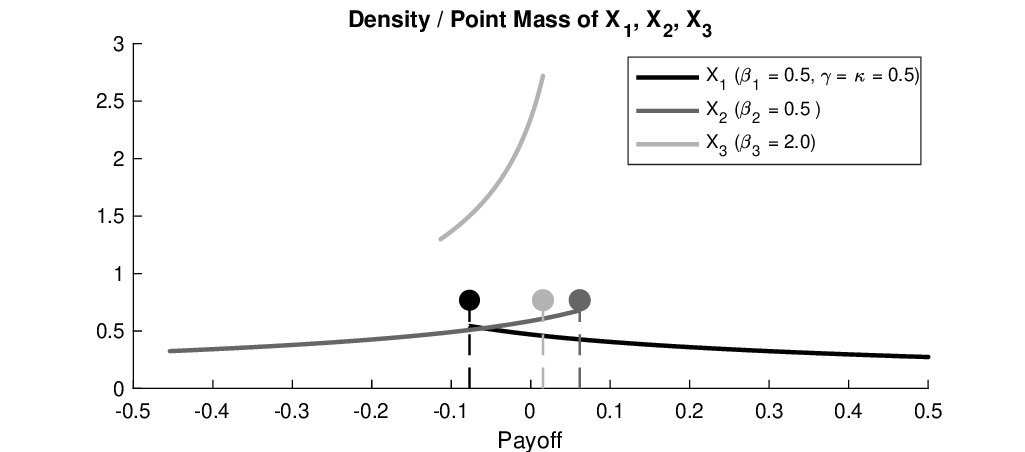}
\caption{\small{ Risk allocations with HEU;  we set $\gamma=\kappa=0.5$, $\beta_1=\beta_2=0.5$ and $\beta_3=2$. }}
\label{fig:HEU2}
\end{figure}

\medskip

For the risk allocation, we provide an example in Figure \ref{fig:HEU2}. In particular, with $\gamma=\kappa=0.5$, we observe that the probability weighting function is inverse S-shaped, and $p^*\approx0.768$. That is, the probability of the full-insurance event is approximately $76.8\%$. Moreover, because $\tilde T'(1-)=1.75$ is finite, we find an abrupt jump in the density function to/from zero in the tail. This is in contrast to Prelec probability weighting functions. {Since $\tilde T'(p^*)\approx 0.95$ and $\tilde T=\delta$ on $[p^*,1]$, it follows that the distribution of  $\ln(\delta'(U))$ is supported on $[-0.05, 0.56]$. By Corollary \ref{Cor_explict} and as displayed in Figure \ref{fig:HEU2}, the continuous density component of $X_1$ has a range of width $(0.56+0.05)/(\beta_1+\bar \beta)\approx0.68$. The continuous density component of $X_2$ has a range of width $\frac{\bar\beta }{\beta_2}0.68=0.54$, and the continuous density component of $X_3$ has the smallest range of width $\frac{\bar\beta }{\beta_3}0.68=0.135$, due to larger risk aversion. }

\medskip
\subsection{Side Payments}

In the context of side payments, the Kaldor-Hicks criterion (e.g., \cite{kaldor1939}) offers a useful efficiency benchmark. An alternative allocation is a potential Pareto improvement  if there exists a feasible transfer scheme that  compensates all agents who lose from the change. In our framework, the surplus gained by agents advantaged through RDU-induced probability distortions could hypothetically finance such compensation, implying that the resulting allocation satisfies the Kaldor-Hicks criterion even without the transfers being implemented.

\medskip

We next display the certainty equivalents for the three agents as a function of $\alpha$ from \eqref{eq:prelec}. We start with the setting of  Figure \ref{fig:1b}.  For Agent 1, the certainty equivalent of any payoff $X$ with cumulative distribution function $F_{X}$ is given by:
$$CE_1(X):=u_1^{-1} \left(\displaystyle\int_0^1 u_1\left(F_{X}^{-1}(p)\right) \tilde T^\prime(p) \, \textnormal{d}p\right),\qquad\quad
$$  
and for the other agents it is given by:
$$CE_i(X):=u_i^{-1} \left(\displaystyle\int_0^1 u_i\left(F_{X}^{-1}(p)\right)  \, \textnormal{d}p     \right),\quad  i=2,\ldots,n.$$  

\medskip

By construction, certainty equivalents satisfy $CE_j(c)=c$, for $c\in\mathbb{R}$. Moreover, for exponential utilities, $CE_j$ is cash additive: $CE_j(X+c)=CE_j(X)+c$, and it is also called an entropic risk measure if $T=id$.  Here, we assume that all initial endowments are 0. By Corollary \ref{corsec2}, if $T=id$, then $(X_1,X_2,X_3)$  is a full-insurance allocation, and we have  $\sum_{i=1}^n CE_i(X_i)=\mathtt{w}=0$, so that a positive certainty equivalent corresponds to a strict improvement over no trading. That is, if $T\neq id$, then $CE=\sum_i CE_i(X_i) > 0$.

\begin{figure}[!htbp]
  \centering
  \begin{subfigure}[b]{0.5\textwidth}
    
\includegraphics[width=8cm,height=6cm,clip]{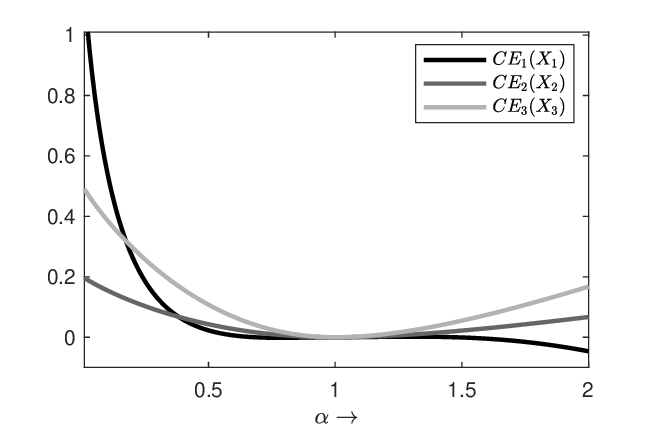}
    \caption{\small The individual certainty equivalents  \newline $CE_i(X_i)$ of the three agents, for  varying $\alpha$.}
    \label{fig:CEQ}
  \end{subfigure}%
  \hfill
  \begin{subfigure}[b]{0.5\textwidth}
    \includegraphics[width=8cm,height=6cm,clip]{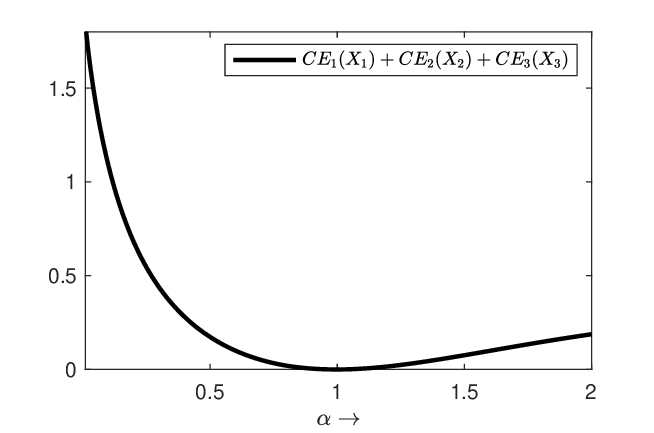}
    \caption{\small The aggregate certainty equivalents\newline\qquad $CE$, for varying $\alpha$.}
    \label{fig:CEQtotal}
  \end{subfigure}
  \caption{Here, we use $\beta_1=\frac{1}{2}=\beta_2$, and $\beta_3=2$.  }
  \label{fig:CEQ-examples}
\end{figure}

\medskip

We show the certainty equivalents for our baseline example with various choices of $\alpha$ in Figure \ref{fig:CEQ}. Notice that the certainty equivalents are equal to 0 when $\alpha=1$. In this case, there are no risk exposures, corresponding to no betting. Just like in \cite{BBG24}, we can argue that zero-sum deterministic cash transfers (side payments) can be arbitrarily selected, and thus a negative value of $CE_1(X_1)$ for the RDU agent 1 when $\alpha>1.5$ can be increased to a positive value via side payments from the other agents. Thus, only the sum of certainty equivalents is relevant, and a positive sum of certainty equivalents means that all agents may be better off from betting in combination with some appropriately selected side payments. We display the sum of certainty equivalents in Figure \ref{fig:CEQtotal}.
We see that for $\alpha$ further away from 1, the sum of certainty equivalents becomes larger, which indicates that betting is attractive due to the probability distortion  of the RDU agent.

\bigskip
\section{Nudging the RDU agent}\label{Sec:nud}

Consider an economy populated by one RDU agent and one EU agent with a given distortion function $T$.\footnote{By Lemma \ref{PropRepAg}, this single EU agent can be interpreted as the representative agent for the $n-1$ EU agents in the market, endowed with the aggregate utility function $u_\lambda$.} As shown in Section \ref{SecBettingPO} and Section \ref{Section:noisePO}, the degree of nonlinearity of the RDU agent's distortion function directly impacts the degree of non-idiosyncratic risk within any efficient allocation.  

\medskip

We assume that a sophisticated social planner is aware of the deterministic aggregate endowment $\mathtt{w}$, and wishes to control the level of non-idiosyncratic risk within efficient allocations resulting from the RDU agent's probability weighting. The social planner can exert a costly effort to nudge the RDU agent, and this nudging brings the RDU agent's probability weighting function closer to linearity. The effort $M$ is exerted by the social planner,  at a monetary cost $M>0$, which reduces the aggregate endowment to $\mathtt{w}-M$. This impacts the curvature of the ``nudged'' weighting function $T_M$.\footnote{One illustrative real-world example is to \emph{educate} an RDU agent in probability theory and statistics. In this perspective, $T$ models a \emph{lack of information}, rather than a probabilistic form of ``risk aversion.'' Meanwhile, EU agents behave much like econometricians, and the social planner aspires to elevate the RDU agent to the same level of expertise. However, this ``education'' or ``nudging'' comes at a cost, which explains why it does not occur automatically. }

\medskip

The welfare maximization problem with nudging then reads as follows:
\begin{eqnarray}
\label{nudge}
\max_{M\in [0,\mathtt{w}]} \ \sup_{X_1+X_2=\mathtt{w}-M} \ \underbrace{ U_1(X_1, T_M)+U_2(X_2)}_{=:W(X_1, M)},
\end{eqnarray}

\noindent with a ``nudged'' weighting function $T_M: [0,1]\to [0,1]$ that depends on the effort level $M$ and is given by 
\begin{eqnarray}\label{TM}
T_M(p):= (1-f(M)) \, T(p)+f(M) \, p.
\end{eqnarray}

\medskip

\noindent We assume that the function $f$ is a strictly increasing and  concave function $f:[0,\mathtt{w}]\to [0,1]$ that satisfies the following:
\medskip
\begin{enumerate}
\item The outer  problem is well-defined. 
\medskip
\item If $M>M'$ then  $T_M$ is more linear than $T_{M'}$.
\medskip
\item $f(0) = 0$, so that $T_0=T$.
\end{enumerate}

\medskip

In the numerical illustration of Example \ref{ex_nudge}, we adopt the specific functional form
$f(M) = 1-(1-M)^k$, for $k>1$, so that higher values of $k$ correspond to stronger curvature, i.e., a faster convergence of $T_M$ toward linearity as $M$ increases.

\medskip

\begin{example} 
We continue studying the case with a Prelec weighting function, $n=3$, and exponential utilities. Thus, we now use two EU agents, which can be summarized as a representative EU agent via Lemma \ref{PropRepAg}.
Consider the case in which $\alpha=0.4$, $\beta_1=\frac{1}{2}
=\beta_2$,  $\beta_3=2$. 
In Figure \ref{fig:convex}, we illustrate how the optimal allocation changes with $M$. Specifically, we plot the atom and the atomless density $\mathtt{f}^M_1$ of $X^M_1$. By Corollary \ref{Cor_explict}, this distribution of $X^M_1$ coincides with the one obtained in the three-agent economy when the two EU agents are aggregated into a single representative EU agent with risk-aversion parameter $\bar \beta=0.4$.

\medskip

\begin{figure}[!htpb]
\includegraphics[width=14cm,height=6cm,clip]{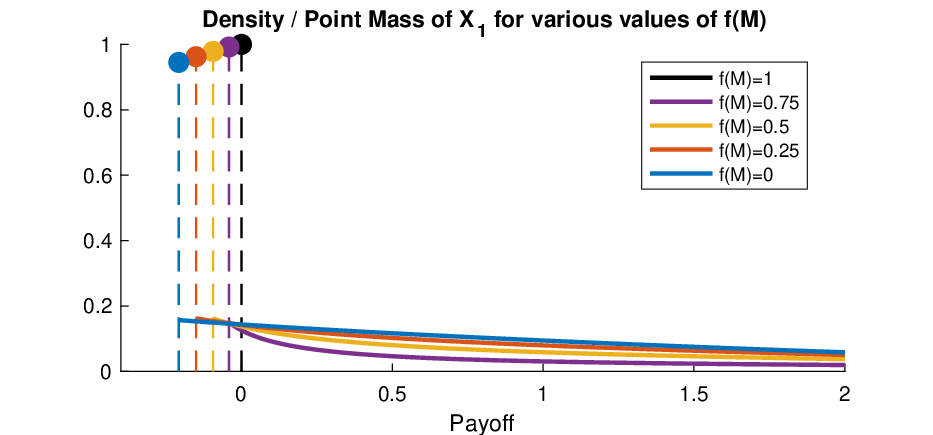}
\caption{\small{Examples of the density $\mathtt{f}^M_1$ for various choices of $f(M)$ under the base case. We set $f(M)=0, 0.25, 0.5, 0.75, 1$.}} 
\label{fig:convex}
\end{figure}

We find that the atom becomes more likely when $f(M)$ increases, and moreover, the densities become lower as $f(M)$ takes lower values. This means that increases in $f(M)$ make the optimal risk allocation ``closer'' to the full insurance allocation.
\end{example}

\medskip

The following result gives a handy representation of the RDU functional after nudging at level $M$.

\medskip

\begin{lemma}\label{L:TM}
The RDU functional with distortion function $T_M$ from \eqref{TM} is given by 
\begin{eqnarray}
U_1(X, T_M)=\left[1 - f(M)\right] \, U_1(X, T) + f(M) \, \int u_1(X) \, \textnormal{d}\p . 
\end{eqnarray}
\end{lemma}

\medskip

These results highlight the functional dependency of the planer's effort level $M$ in Problem   \eqref{nudge}  with value function $M\mapsto V(M)$. We can apply all insights from Sections 2-4.
In particular $V(M)$ can be analyzed in a pointwise manner. In view of Lemma \ref{L:TM}, the social planner is assigning a total welfare to the RDU agent that consists of two parts: (i) a fraction $1-f(M)$ coming from the prior preferences of the RDU agent; and (ii) a fraction $f(M)$  coming from discarding the probability weighting of the RDU agent and treating this agent as an EU-maximizer. 

\medskip

In this case, the characterization in the economy with an RDU having distortion  $T_M$  follows again by Theorem \ref{ThMain}. However,  we need to replace $\delta$ by $\delta_M$,  the convex envelope of $\tilde T_M$.
The following result shows that the modified convexification depends in an affine linear way on $f$. 

\medskip

\begin{lemma}\label{L:delta}
For $T_M$ as in \eqref{TM},  we have  $\delta_M\(p\) = \delta(p)- f(M) (\delta\(p\)-p)$.
\end{lemma}

\medskip

It directly follows  that $ \delta'_M\(p\)=\delta'(p) -f(M)(\delta'(p)-1)$. 
Lemma \ref{L:delta} gives a useful decomposition of $\delta_M$. In view of the optimal sharing rule under $T_M$, the resulting   properties of  $M\mapsto \delta_M'$ clarify  the relationship between $\delta^\prime$ and $\delta^\prime_M$. 
For instance,  if $T$ is inverse S-shaped, then by Proposition \ref{Prop:MG}(1)  there exists a unique tangent point $p_M^* \in \(0,\bar p\)$ such that $\delta_M$ is given by
\begin{equation*}
\delta_M\(p\) = \left\{
\begin{array}{l l}
\frac{\tilde T_M\(p^*_M\)}{p^*_M} \cdot p & \quad \mbox{if $p < p^*_M$;}\\
\tilde T_M\(p\) & \quad \mbox{otherwise. }\\
\end{array} \right.
\end{equation*}
If  $T$ is S-shaped,  the argument is similar (see Proposition \ref{Prop:MG}(2)).

\medskip

We now formulate for every $M$ the optimal allocation in closed form.
We focus first on the inner part of Problem \eqref{nudge}, the weighted-sum risk-allocation problem under given $M$:
\begin{equation}\label{eqgen-prob2}
V(M)=\sup_{X_1+X_2=\mathtt{w}-M} \, \Big\{U_1(X_{1},T_M)+ \lambda_2 \, U_2(X_{2})\Big\}, \ \quad \lambda_2 > 0.
\end{equation}
Moreover, we show the (smooth) sensitivity of the optimal solution with respect to the effort level $M$ in closed form.

\medskip

\begin{proposition}\label{ThMain2}
\

\begin{enumerate}
\item  For a given effort level $M$,  $\mathbf{X}^M=(x_M(\mathtt{U}) , \mathtt{w}-M-x_M(\mathtt{U} ))$ solves Problem \eqref{eqgen-prob2} if and only if
\begin{align*}
x_M(\mathtt{U}) =m^{-1}_{\lambda_2}\Big(  \delta'(\mathtt{U}) -f(M)(\delta'(\mathtt{U})-1) \Big), 
\end{align*}
where $m_{\lambda_2}(x):=\lambda_2 \, \frac{u_2^\prime\(\mathtt{w}-M-x\)}{u_1^\prime\(x\)}$.

\bigskip

\item Moreover, the optimal allocation depends smoothly on the effort level $M$:
\begin{equation}
\label{eqxprime}
\begin{split}
x'_p(M)
&:=\frac{\partial }{\partial M}X_M(p)\\
&=
\frac{-f^\prime(M)\left(\delta^\prime(p)-1\right)
+\lambda_2\,\dfrac{u_2^{\prime\prime}\left(\mathtt{w}-M-m_{\lambda_2}^{-1}\left(\delta^\prime(p)-f(M)\left(\delta^\prime(p)-1\right)\right)\right)}{u_1^\prime\left(m_{\lambda_2}^{-1}\left(\delta^\prime(p)-f(M)\left(\delta^\prime(p)-1\right)\right)\right)}}
{\Lambda\left(\delta^\prime(p)-f(M)\left(\delta^\prime(p)-1\right)\right)},
\end{split}
\end{equation}

\noindent where
\begin{align*}
\Lambda(x)
&:=\frac{d}{dx}m_{\lambda_2}\,\big(m_{\lambda_2}^{-1}(x)\big)\\
&=\frac{- \lambda_2}{u_1'\left(m_{\lambda_2}^{-1}(x)\right)}
\left[
u_2''\left(\mathtt{w} - M - m_{\lambda_2}^{-1}(x)\right)
+
u_2'\left(\mathtt{w} - M - m_{\lambda_2}^{-1}(x)\right)
\frac{u_1''\left(m_{\lambda_2}^{-1}(x)\right)}
{u_1'\left(m_{\lambda_2}^{-1}(x)\right)}
\right].
\end{align*}
\end{enumerate}
\end{proposition}

\bigskip

Since $\Lambda(x)\geq 0$ and the second summand in the numerator of \eqref{eqxprime} is negative, Proposition \ref{ThMain2}(2) implies that the mapping $M\mapsto x_M(p)$ is increasing only for those values  of $p$ for which the term  $f^\prime(M)\left(\delta^\prime(p)-1\right) 
$ is sufficiently negative, i.e., when $\delta'(p)<1$, and $f'(M)>0$ is large. By concavity of the function $f$, this happens for small values of $M$. This is consistent with Figure \ref{fig:convex}, 
which shows that the  density components of $X_1^M =x_M(\mathtt{U})$ (for varying $M$)  are predominantly  decreasing in $M$ in the right tails, and they increase only in a small region around the corresponding  atoms.

\medskip

Proposition \ref{ThMain2}(2) quantifies the sensitivity of  $X_M$ to changes in $M$ and  will also  help in determining the optimal  level of effort. This is the content of the following result. 
{Without loss of generality, we set $\lambda_2=1$.}

\medskip

\begin{theorem}\label{welfareTHM}
The optimal level of effort $M^*=M^*(T,f)$ can be characterized through FOC of \eqref{nudge}. That is,
\begin{eqnarray*}
W(X_{M^*},M^*) = \max_{M} \, W(X_M,M) = \max_{M} \, \max_{X} \, W(X,M),
\end{eqnarray*}
 where  $M^*$ solves the equation 
$$\int_0^1 \frac{\partial}{\partial M} \, \left[u_1(x_M(t)) \, \tilde T^\prime_M(t) \right] \, dt = - \int_0^1  \frac{\partial}{\partial M} \, u_2(\mathtt{w}-M - x_M(1-t))  \, dt.$$
\end{theorem}

\bigskip

In the following example, we illustrate this result by returning to the setup of Section \ref{Section:noisePO.2}, with two CARA agents. 

\medskip

\begin{example} \label{ex_nudge}
Let $\mathtt{w}=1$ and $\lambda_2=1$.  The RDU agent has a distortion function $T(p) =\exp(-(-\ln(p))^{\alpha})$, with conjugate distortion function  $\tilde T(p) = 1- \exp(-(-\ln(1-p))^{\alpha})$.
The  two agents have CARA utility functions $u_i(x) = -\frac{1}{\beta_i} e^{-\beta_i x}$, with first derivatives $u_i'(x) = e^{-\beta_i x}$, for $i=1,2$. In the numerical illustration we set $\beta_1 = \tfrac{1}{2}$ and $\beta_2 = \bar\beta = 0.4$, so that the EU agent coincides with the representative agent from Corollary \ref{Cor_explict}.  Using $\tilde T^\prime_M(t)=f(M) + (1-f(M)) \, \tilde T^\prime(t)$ and $\frac{\partial}{\partial M} \, \tilde T^\prime_M(t)=f^\prime(M) \, (1 - \tilde T^\prime(t))$, it follows that the  FOC from Theorem \ref{welfareTHM} is given by
\begin{equation}
\begin{split}
\label{inteq} 
&\displaystyle\int_0^1 e^{-\beta_1{x_M(t)} } \left[ x_M^\prime(t) \, \Big(f(M) + (1-f(M)) \, \tilde T^\prime(t)\Big) - \frac{1}{\beta _1} f'(M)( \tilde T^\prime(t) -1)\right]dt \\
&\qquad= \displaystyle\int_0^1 e^{-\beta_2(1 - M- x_M(1-t))}  \left[1 + x_M^\prime(1-t)\right]  dt.
\end{split}
\end{equation}

\medskip

In the following, we aim to  derive the function $M\mapsto X_M$ in closed form, by applying  Proposition \ref{ThMain2}(1). With $\beta_1+\beta_2=1$,  we obtain first 
$$
m(x)
=\frac{u_2'(\mathtt{w}-M-x)}{u_1'(x)}
=e^{-\beta_2(\mathtt{w}-M-x)+\beta_1 x}
=e^{-\beta_2(\mathtt{w}-M) + (\beta_1 + \beta_2) x},
$$

\noindent with inverse function $m^{-1}(x)=\frac{\ln(x)+\beta_2(\mathtt{w}-M)}{\beta_1+\beta_2} = \ln(x) + \beta_2 (1-M)$. Substituting $m^{-1}(x)$  into  $X_M$   from Proposition \ref{ThMain2}(1) yields 
$$x_M(t) = \ln \left( \delta'(t) - f(M)\left(\delta'(t) - 1\right) \right) + \beta_2(1 - M).$$

\medskip

We now compute $X_M'$. Based on Proposition \ref{ThMain2}(2), we aim  to find $\Lambda(x) =\frac{d}{dx}m\,\big(m^{-1}(x)\big)$ in closed form. Now, letting $y:=m^{-1}(x)$, we have
\begin{align*}
\Lambda(x)
&=\frac{- \lambda_2}{u_1'\left(m^{-1}(x)\right)}
\left[
u_2''\left(\mathtt{w} - M - m^{-1}(x)\right)
+
u_2'\left(\mathtt{w} - M - m^{-1}(x)\right)
\frac{u_1''\left(m_{\lambda_2}^{-1}(x)\right)}
{u_1'\left(m_{\lambda_2}^{-1}(x)\right)}
\right]\\
&=\frac{-\lambda_2}{e^{-\beta_1 y}}
\left[
-\beta_2 e^{-\beta_2(\mathtt{w}-M-y)}
+
e^{-\beta_2(\mathtt{w}-M-y)}\,
\frac{-\beta_1 e^{-\beta_1 y}}{e^{-\beta_1 y}}
\right]
=\frac{-\lambda_2}{e^{-\beta_1 y}}
\left[
-(\beta_1+\beta_2)\,e^{-\beta_2(\mathtt{w}-M-y)}
\right]\\
&=\lambda_2(\beta_1+\beta_2)\,
e^{-\beta_2(\mathtt{w}-M)+(\beta_1+\beta_2)\,y}
=e^{-\beta_2(1-M)+ y}
=e^{-\beta_2(1-M)+\ln(x) + \beta_2 (1-M)}
=x.
\end{align*}

\medskip

\noindent Consequently, for each $t$, 
\begin{align*}
x'_t(M)
&:=\frac{\partial }{\partial M}X_M(t)
=
\frac{-f^\prime(M)\left(\delta^\prime(t)-1\right)
+\lambda_2\,\dfrac{u_2^{\prime\prime}\left(\mathtt{w}-M-m^{-1}\left(\delta^\prime(t)-f(M)\left(\delta^\prime(t)-1\right)\right)\right)}{u_1^\prime\left(m^{-1}\left(\delta^\prime(t)-f(M)\left(\delta^\prime(t)-1\right)\right)\right)}}
{\Lambda\left(\delta^\prime(t)-f(M)\left(\delta^\prime(t)-1\right)\right)}\\
&=
\frac{-f^\prime(M)\left(\delta^\prime(t)-1\right)
+\dfrac{u_2^{\prime\prime}\left(\mathtt{w}-M-m^{-1}\left(\delta^\prime(t)-f(M)\left(\delta^\prime(t)-1\right)\right)\right)}{u_1^\prime\left(m^{-1}\left(\delta^\prime(t)-f(M)\left(\delta^\prime(t)-1\right)\right)\right)}}
{\delta^\prime(t)-f(M)\left(\delta^\prime(t)-1\right)}\\
&=\frac{
- f'(M)\left(\delta'(t)-1\right)
-
\left[\delta'(t)-f(M)(\delta'(t)-1)\right] \beta_2
}
{
\delta'(t)-f(M)(1-\delta'(t))
}
=\frac{f'(M)\left(1-\delta'(t) \right)}
{{(1-f(M))\delta'(t)+f(M)}}- \beta_2.
\end{align*}

\bigskip

Substituting $x_M$ and $x_M'$ into \eqref{inteq} gives an integral equation to identify the optimal effort level $M^*(k,\alpha)\in (0,1)$ that maximizes the social welfare of the nudging planner. If we now take the specific form  $f(M):=1-(1-M)^k$, with $k>1$, then in view of Section \ref{Section:noisePO}, we see that $M^*$ reduces the riskiness of the allocation, since the full-insurance term $\mathbb{FI}$ increases with $M^*$, at the cost of decreasing aggregate wealth $1-M^*$. 

\medskip

With sufficient curvature in $f$, say $k=20$ and $\alpha=0.4$, the integral equation in \eqref{inteq} can be solved numerically. This gives $M^* \approx 6.57\%$ of the aggregate wealth $\mathtt{w}=1$ as the optimal fraction of the (constant) aggregate endowment that is invested in nudging. However, a distortion with less linearity, say $\alpha=0.2$, results in a  significant increase in  the  optimal fraction to $M^* \approx 9.25\%$. See Figure \ref{fig:ex56} for an overview.  

\begin{figure}[!htpb]
\includegraphics[width=12cm,height=6cm,clip]{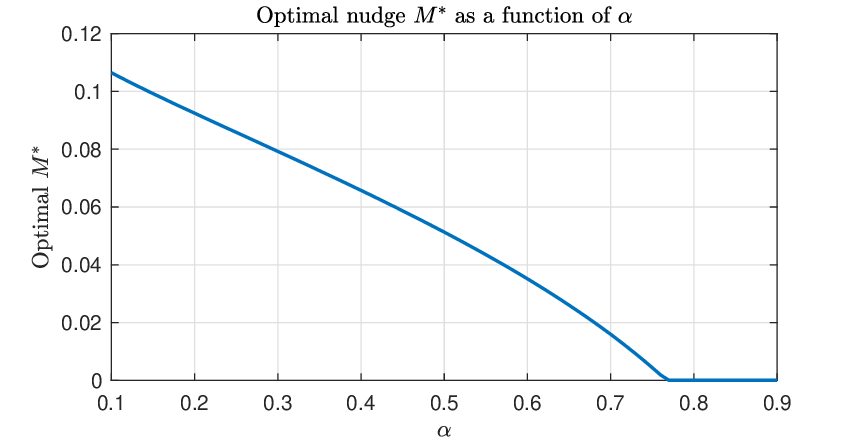}
\caption{\small{Optimal investment in nudging as a function of the Prelec parameter $\alpha$. For $\alpha> 0.76$, we have $M ^*=0$.}} 
\label{fig:ex56}
\end{figure}
\end{example}

\medskip

Figure \ref{fig:ex56} shows that the optimal effort $M^*$ is zero whenever the Prelec parameter $\alpha$ is sufficiently close to its EU benchmark value $\alpha=1$. In this region, the probability weighting function is almost linear, so the welfare gains from correcting the internality are too small to justify the reduction in aggregate wealth, and the planner will not intervene. Only once probability distortions become sufficiently pronounced (that is, when $\alpha$ is much smaller than $1$) is $M^*$ strictly positive. This illustrates that the social planner is willing to tolerate a nontrivial amount of endogenous uncertainty in equilibrium, before spending resources on nudging the RDU agent toward more linear probability weighting.

\bigskip
\section{Conclusion}\label{sec:6}

In pure-exchange economies with no aggregate uncertainty populated by Subjective-Expected Utility (SEU) maximizers, Pareto-efficient allocations are no-betting (full-insurance) allocations if and only if the agents have common beliefs. In this paper, we show how the introduction of a single Rank-Dependent Utility (RDU) agent, into an otherwise classical setting with EU agents that have common beliefs, can lead to drastically different predictions. We demonstrate how betting, that is, uncertainty-generating trade, can be Pareto improving despite the presence of a common baseline probability measure on the state space. That is, probability weighting endogenously generates betting at an optimum, even under common baseline beliefs, thereby showing that uncertainty-generating trade can arise purely from heterogeneity in the perception of risk, rather than in beliefs. 

\bigskip

Our analysis provides a micro-founded explanation for the coexistence of common beliefs and speculative behavior in an environment with no initial aggregate uncertainty. We quantify this behavior by providing a crisp closed-form characterization of Pareto optima that shows precisely how betting behavior emerges at optima. Additionally, we provide comparative statics for some popular classes of probability weighting functions. Finally, we examine how an appropriately designed intervention, such as statistical or financial education of the RDU agent, can attenuate undesirable speculative trade and partially re-establish the optimality of full-insurance allocations.

\newpage

\begin{appendices}
\hypertarget{LinkToAppendix}{\ }
\setlength{\parskip}{0.5ex}
\vspace{-0.5cm}

\section{Proofs for Section \ref{SecBettingPO}}
\label{AppProofs}

\medskip

By a slight abuse of notation, set $\displaystyle\int \cdot \,  \textnormal{d}t =\int_0^1 \cdot \, \textnormal{d}t$ and  $\displaystyle\int \cdot \, \textnormal{d}\p=\int_\Omega \cdot \, \textnormal{d}\p$.

\bigskip

\noindent{\bf{Proof of Proposition \ref{PropPOMax}:}} (1) Immediate. (2) When $T$ is convex, it follows from \cite{Chewetal1987} that $U_1$ is concave. Moreover, by assumption, the functionals $U_i$ are concave, for each $i \in \{2, \ldots, n\}$. The rest follows from a classical separation argument \cite[Proposition 3.4]{CarlierDana2006a}.\qed

\bigskip

\noindent{\bf{Proof of Proposition \ref{PropComChar}:}}  \ 

{$(1) \implies (2)$:} 
Fix $\mathbf{X}_{-1}\in{\mathcal A}_{-1}^C(\mathtt{w}-X)$. Then by a classical result \cite[Lemma 4.95]{FollmerSchied2025}, 
there are nondecreasing and $1$-Lipschitz functions $g_i$ that sum to the identity function, such that for each $i \in \{2, \ldots, n\}$, $X_i = g_i(\mathtt{w}-X)$. Since the probability space is nonatomic, it follows from \cite[Lemma D.17]{FollmerSchied2025} that there exists a random variable $\mathtt{U}\sim Uni(0,1)$  such that $\mathtt{w}-{X} = F_{\mathtt{w}-{X}}^{-1}(\mathtt{U})$, a.s. In particular, $\mathtt{w}-{X}$ is comonotonic with $\mathtt{U}$. Moreover, for each $i \in \{2, \ldots, n\}$, 
$$
X_i = g_i(\mathtt{w}-X) 
= g_i\left(F_{\mathtt{w}-{X}}^{-1}(\mathtt{U})\right)
= F_{g_i(\mathtt{w}-{X})}^{-1}(\mathtt{U})
=F_{X_i}^{-1}(\mathtt{U}), \ \hbox{a.s.,}
$$
by monotonicity of each function $g_i$ \cite[Lemma D.12]{FollmerSchied2025}.

\medskip

{$(2) \implies (1)$:} 
Suppose that there exists some $\mathtt{U}\sim Uni(0,1)$ comonotonic with $\mathtt{w}-X$, such that $X_i = F_{X_i}^{-1}(\mathtt{U})$, a.s., for each $i \in \{2, \ldots, n\}$. Then clearly the vector $\(X_2, \ldots, X_n\)$ is comonotonic.\qed

\bigskip

\noindent{\bf{Proof of Proposition \ref{PropInnerQuant}:}} 
By nonatomicity of the space, let $\mathtt{U}\sim Uni(0,1)$ be such that $\mathtt{w}-{X} = F_{\mathtt{w}-{X}}^{-1}(\mathtt{U})$, a.s. Let $\{f_i^*\}_{i=2}^n$ be optimal for \eqref{InnerProbCom}. For each $i \in \{2, \ldots, n\}$, let $X_i^*:= f_i^*(\mathtt{U})$. Then $\{X_i^*\}_{i=2}^n$ is comonotonic since $f_i^*$ is nondecreasing for all $i \in \{2, \ldots, n\}$. Moreover, by feasibility of $\{f_i^*\}_{i=2}^n$ for \eqref{InnerProbCom},
$$\sum_{i=2}^n X_i^*:= \sum_{i=2}^n f_i^*(\mathtt{U}) = F_{\mathtt{w} - {X}}^{-1}(\mathtt{U}) = \mathtt{w}-{X}, \ a.s.$$
Thus, $(X_i^*)_{i=2}^n\in{\mathcal A}^C_{-1}(\mathtt{w}-{X})$, and hence it is feasible for \eqref{InnerProb}. Suppose that $\{X_i^*\}_{i=2}^n$ is not optimal for \eqref{InnerProb}. Then there exists some $(Z_i)_{i=2}^n\in{\mathcal A}^C_{-1}(\mathtt{w}-{X})$ such that
$$\sum_{i=2}^n\lambda_i \, U_i(Z_i) > \sum_{i=2}^n\lambda_i \, U_i(X_i^*).$$
For each $i \in \{2, \ldots, n\}$, let $g_i := F_{Z_i}^{-1}$. Then since $(Z_i)_{i=2}^n\in{\mathcal A}^C_{-1}(\mathtt{w}-{X})$,
$$\sum_{i=2}^n g_i= F_{\sum_{i=2}^n Z_i}^{-1} =  F_{\mathtt{w} - {X}}^{-1}.$$
Thus $\{g_i\}_{i=2}^n$ is feasible for \eqref{InnerProbCom}, and so
$$\sum_{i=2}^n \lambda_i \, U_i(Z_i) = \int \sum_{i=2}^n \lambda_i \, u_i\left(g_i(t)\right) \,  \textnormal{d}t \leq \int \sum_{i=2}^n \lambda_i \, u_i\left(f_i^*(t)\right) \,  \textnormal{d}t = \sum_{i=2}^n \lambda_i \, U_i(X_i^*),$$
a contradiction. Therefore, $\{X_i^*\}_{i=2}^n$ is optimal for \eqref{InnerProb}.

\vspace{0.3cm}

Conversely, let $\{X_i^* := f_i^*(\mathtt{U})\}_{i=2}^n$ be optimal for \eqref{InnerProb}. Since $(X_i^*)_{i=2}^n\in{\mathcal A}^C_{-1}(\mathtt{w}-{X})$, it follows that
$$\sum_{i=2}^n f_i^* = \sum_{i=2}^n F_{X_i^*}^{-1} = F_{\sum_{i=2}^n X_i^*}^{-1} = F_{\mathtt{w}-{X}}^{-1},$$
and hence $\{f_i^*\}_{i=2}^n$ is feasible for \eqref{InnerProbCom}. Suppose that $\{f_i^*\}_{i=2}^n$ is not optimal for \eqref{InnerProbCom}. Then there exists some $(g_i)_{i=2}^n \in \mathcal{Q}$ such that $\sum_{i=2}^n g_i =  F_{\mathtt{w} - {X}}^{-1}$ and
$$\int \sum_{i=2}^n \lambda_i \, u_i\left(g_i(t)\right) \,  \textnormal{d}t > \int \sum_{i=2}^n \lambda_i \, u_i\left(f^*_i(t)\right) \,  \textnormal{d}t = \int \sum_{i=2}^n \lambda_i \, u_i\left(F^{-1}_{X^*_i}(\mathtt{U})\right) \,  \textnormal{d}\mathbb{P}
= \sum_{i=2}^n \lambda_i \, U_i(X_i^*).$$
Letting $Z_i := g_i(\mathtt{U})$ for each $i \in \{2, \ldots, n\}$, it follows that
{\small{$$\sum_{i=2}^n \lambda_i \, U_i(Z_i)
= \int \sum_{i=2}^n \lambda_i \, u_i\left(F^{-1}_{Z_i}(\mathtt{U})\right) \,  \textnormal{d}\mathbb{P}
= \int \sum_{i=2}^n \lambda_i \, u_i\left(g_i(\mathtt{U})\right) \,  \textnormal{d}\mathbb{P}
=\int \sum_{i=2}^n \lambda_i \, u_i\left(g_i(t)\right) \,  \textnormal{d}t
> \sum_{i=2}^n \lambda_i \, U_i(X_i^*),$$}}a contradiction. Hence, $\{f_i^*\}_{i=2}^n$ is optimal for \eqref{InnerProbCom}.\qed

\bigskip
    

\noindent{\bf{Proof of Lemma \ref{PropRepAg}:}} Note that for all $x$,
$$x = I_\lambda(J_\lambda(x)) = \sum_{i=2}^n I_i\left(\frac{J_\lambda(x)}{\lambda_i}\right) \ \ \hbox{and} \ \ u^\prime_\lambda(x) = J_\lambda(x).$$

\medskip

\noindent Hence, it follows that $ f^\diamond_2,\ldots, f^\diamond_n \in \mathcal{Q}$ and
$\sum_{i=2}^n  f^\diamond_i = \sum_{i=2}^n I_i\left(\lambda_i^{-1}J_\lambda\left(F^{-1}_{\mathtt{w}-{X}}\right)\right) = F^{-1}_{\mathtt{w}-{X}}.$
  That is, $\{ f^\diamond_i\}_{i=2}^n$ is feasible for \eqref{InnerProbCom}. In addition, if $\{f^*_i\}_{i=2}^n$ is optimal for \eqref{InnerProbCom}, then
$$\sum_{i=2}^n  \(f^*_i(t)- f^\diamond_i(t)\) = 0, \, a.s.,$$

\begin{equation}\label{eqleq}
\textnormal{ and }\int u_\lambda\(F^{-1}_{\mathtt{w}-{X}}(t)\)  \, \textnormal{d}t
= \int \sum_{i=2}^n \lambda_i \, u_i\( f^\diamond_i(t)\) \, \textnormal{d}t
\leq  \int \sum_{i=2}^n \lambda_i \, u_i\(f^*_i(t)\) \, \textnormal{d}t.
\end{equation}

\medskip

\noindent Moreover, for a fixed $t \in [0,1]$, the concavity of $u_i$ for $i \in \{2, \ldots, n\}$ yields
\begin{align*}
\sum_{i=2}^n \lambda_i \, u_i\(f^*_i(t)\)
&\leq \sum_{i=2}^n \lambda_i \, u_i\( f^\diamond_i(t)\) +  \sum_{i=2}^n \lambda_i \, \(f^*_i(t)- f^\diamond_i(t)\) \, u_i^\prime\( f^\diamond_i(t)\) \\
&= u_\lambda\(F^{-1}_{\mathtt{w}-{X}}(t)\)  +  \sum_{i=2}^n \lambda_i \, \(f^*_i(t)- f^\diamond_i(t)\) \, \underbrace{u_i^\prime\Big(I_i\Big(\frac{J_\lambda(F^{-1}_{\mathtt{w}-{X}}(t))}{\lambda_i}\Big)\Big)}_{= \lambda_i^{-1} J_\lambda\left(F^{-1}_{\mathtt{w}-{X}}(t)\right)}.
\end{align*}

\medskip

\noindent Hence, for a.e.\ $t \in [0,1]$, the derivation yields $\sum_{i=2}^n \lambda_i \, u_i\(f^*_i(t)\)
\leq u_\lambda\(F^{-1}_{\mathtt{w}-{X}}(t)\).$ Therefore,
\begin{equation}\label{eqgeq}
\int   \sum_{i=2}^n \lambda_i \, u_i\(f^*_i(t)\) \,  \textnormal{d}t \leq  \int  u_\lambda\(F^{-1}_{\mathtt{w}-{X}}(t)\) \textnormal{d}t.
\end{equation}

\medskip

\noindent Consequently, by \eqref{eqleq} and \eqref{eqgeq},
\begin{equation*}
\int \sum_{i=2}^n \lambda_i \, u_i\(f^*_i(t)\) \, \textnormal{d}t
=  \int u_\lambda\(F^{-1}_{\mathtt{w}-{X}}(t)\) \, \textnormal{d}t
= \int \sum_{i=2}^n \lambda_i \, u_i\( f^\diamond_i(t)\) \, \textnormal{d}t,
\end{equation*}

\noindent implying that $\{ f^\diamond_i\}_{i=2}^n$ is optimal for \eqref{InnerProbCom}. Additionally, note that
\begin{align*}
\sup_{\mathbf{X}_{-1}\in{\mathcal A}_{-1}(\mathtt{w}-{X})}\ \sum_{i=2}^n\lambda_i \, U_i(X_i)
&= \sum_{i=2}^n\lambda_i \, U_i\( f^\diamond_i(\mathtt{U})\)
= \sum_{i=2}^n\lambda_i \, \int u_i\( f^\diamond_i(\mathtt{U})\)  \textnormal{d}\p\\
&= \int \sum_{i=2}^n \lambda_i \, u_i\( f^\diamond_i(t)\) \, \textnormal{d}t
=  \int u_\lambda\(F^{-1}_{\mathtt{w}-{X}}(t)\) \, \textnormal{d}t\\
&=  \int u_\lambda\(F^{-1}_{\mathtt{w}-{X}}(\mathtt{U})\) \, \textnormal{d}\p
=  \int u_\lambda\(\mathtt{w}-{X}\) \,  \textnormal{d}\p.
\end{align*}
\qed

\medskip

As a preparation of the  proof of Lemma \ref{LemBBGnew}, define the pointwise aggregate utility of the EU agents by
\begin{equation}
\label{eqAggUt}
u_\lambda(x) := \sum_{i=2}^n \lambda_i \, u_i\left(I_i\left(\tfrac{J_\lambda(x)}{\lambda_i}\right)\right).
\end{equation}
It also follows from Lemma \ref{PropRepAg} that if $\mathtt{U}  \sim Uni(0,1)$ is such that $\mathtt{w}-{X} = F_{\mathtt{w}-{X}}^{-1}(\mathtt{U})$, a.s., then

\begin{enumerate}
\item For all $i \in \{2, \ldots, n\}$, we have $u_i^\prime( f^\diamond_i(\mathtt{U})) = \lambda_i^{-1} \, u_\lambda^\prime(\mathtt{w}-{X}), \, a.s.$

\medskip

\item $\{ f^\diamond_i(\mathtt{U})\}_{i=2}^n$ is optimal for \eqref{InnerProb} and hence $(\mathtt{w}-{X})$-PO, and $ f^\diamond_i(\mathtt{U}) = I_i\(\frac{J_\lambda\(\mathtt{w}-{X}\)}{\lambda_i}\)$, a.s., for each $i \in \{2, \ldots, n\}$. Moreover, with $U_\lambda(\mathtt{w}-{X})$ as in \eqref{InnerProb}, we have
\begin{align*}
U_\lambda(\mathtt{w}-{X})
= \sum_{i=2}^n\lambda_i \, U_i\( f^\diamond_i(\mathtt{U})\)
=  \int u_\lambda(\mathtt{w}-{X}) \, \textnormal{d}\p.
\end{align*}
\end{enumerate}

\bigskip

\noindent{\bf{Proof of Lemma \ref{LemBBGnew}:}}
Recall that Problem \eqref{eq:gen-prob2} is given by $\sup_{X\in \cX} \left[U_1(X)+U_\lambda(\mathtt{w}-X)\right]$, where by Lemma \ref{PropRepAg}, 
$$
U_\lambda(\mathtt{w}-X) = \int u_\lambda(\mathtt{w}-{X}) \, \textnormal{d}\p,$$
and $u_\lambda$ is defined in \eqref{eqAggUt}. Hence, $X^*$ is optimal for Problem \eqref{eq:gen-prob2} if and only if it is optimal for 
$$\sup_{X\in \cX} \, \left[\int u_1(X) \, \textnormal{d}T\circ\p \ + \int u_\lambda(\mathtt{w}-{X}) \, \textnormal{d}\p\right].$$

\noindent Moreover, as in \cite[Lemma A.1]{BBG24}, 
$$\int u_1(X) \, \textnormal{d}T\circ\p = \int u_1(f_t) \, \tilde T^{\prime}_t \, \textnormal{d}t \ \ \hbox{and} \ \ \int u_\lambda(\mathtt{w}-{X}) \, \textnormal{d}\p = \int u_\lambda\left(\mathtt{w}-f_t\right) \, \textnormal{d}t,$$
where $f_t := F_{X}^{-1}(t)$ and $\tilde T^{\prime}_t := \tilde T^{\prime}(t)$, for all $t \in [0,1]$. Hence, using a quantile reformulation approach as in \cite{BBG24}, $X^*$ is optimal for Problem \eqref{eq:gen-prob2} if and only if it is optimal for 
\begin{align}
\sup_{f \in \mathcal{Q}} \, \int \left[u_1(f_t) \, \tilde  T^{\prime}_t + u_\lambda(\mathtt{w}-f_t) \right] \, \textnormal{d}t.
\label{eq:initial-problem}
\end{align} 

\medskip

We solve Problem \eqref{eq:initial-problem} using a pointwise optimization approach. First, it is easy to verify that, for each $t \in [0,1]$,
$$\bar f_t := m^{-1}_\lambda\left(\tilde T^\prime_t\right) = \argmax_{y} \big\{u_1(y) \, \tilde T^\prime_t + u_\lambda(\mathtt{w}-y) \big\},$$ 
where $m_\lambda(x):=\frac{u_\lambda^\prime(\mathtt{w}-x)}{u_1^\prime(x)}$, for all $x \in\mathbb{R}$. Since $\tilde T^\prime$ might fail to be  monotone, $\bar f$ might fail to be monotone and hence might not be an element of $\mathcal{Q}$. To overcome this difficulty (arising from the nonconvexity of $T$), we consider the following relaxation of Problem \eqref{eq:gen-prob2}:
\begin{eqnarray}
&\sup_{f \in \mathcal{Q}}  \ \int \left[ u_1(f_t)\delta^{\prime}_t  +  u_\lambda(\mathtt{w}-f_t) \right]\,\textnormal{d}t,
\label{prob3alt2INNBelHomMod00}
\end{eqnarray}

\noindent where the function $\delta$ is the (smooth) convex envelope of  $\tilde T$. The convexity of $\delta$ yields monotonicity of $\delta^\prime$, which guarantees that the pointwise optimizer $f^*$ of Problem \eqref{prob3alt2INNBelHomMod00}, given by
$$f^*_t := m^{-1}_\lambda\left(\delta^\prime_t\right) = \argmax_{y} \big\{u_1(y)\delta^{\prime}_t + u_\lambda(\mathtt{w}-y)\big\}, $$ is indeed a quantile function. 

\medskip

Now, for any $f \in \mathcal{Q}$, it follows from Lemma \ref{fctQuantSolLemmaIneqAltINNBelHom2} that
\begingroup
\allowdisplaybreaks
\begin{align*}
\int \left[ u_1(f_t)\tilde T^\prime_t + u_\lambda(\mathtt{w}-f_t) \right] \textnormal{d}t
\leq 
\int u_1(f_t) \left[ \delta^\prime_t + u_\lambda(\mathtt{w}-f_t) \right] \textnormal{d}t
\leq \int \left[ u_1(f^ {*}_t)\delta^{\prime}_t + u_\lambda(\mathtt{w}-f^{*}_t) \right]\textnormal{d}t.
\end{align*}
\endgroup

\noindent Letting $\mathcal{D} := \Big\{t \in \left[0,1\right]: \delta_t \neq \tilde T_t\Big\} = \Big\{t \in \left[0,1\right]: \delta_t < \tilde T_t\Big\}$, it follows that
\begin{equation*}
\begin{split}
\int \left[\tilde T_t - \delta_t\right] \textnormal{d} u_1(f^{*}_t)
= \int_{\mathcal{D}} \left[\tilde T_t - \delta_t\right] \textnormal{d} u_1(f^{*}_t).
\end{split}
\end{equation*}
But, since $\delta$ is affine on $\mathcal{D}$, $\textnormal{d}\delta^{\prime} = 0$ on $\mathcal{D}$, and it follows from $\textnormal{d} f^ {*}_t = \left(m^{-1}_\lambda\right)^{\prime}\left(\delta^{\prime}_t\right) \textnormal{d}\delta^{\prime}_t$ that $\textnormal{d} f^{*}_t= 0$ on $\mathcal{D}$. Consequently, $\displaystyle\int \left[\tilde T_t - \delta_t\right] \textnormal{d} u_1(f^{*}_t) = 0$. Therefore, applying Fubini's Theorem yields
$$0= \int (\tilde T_t - \delta_t) \, \textnormal{d} u_1(f^{*}_t)
=\int u_1(f^{*}_t)( \tilde T^{\prime}_t - \delta^{\prime}_t)  \,  \textnormal{d}t.$$

\noindent Hence, $\displaystyle\int u_1(f^{*}_t) \,  \tilde T^\prime_t \, \textnormal{d}t = \int u_1(f^{*}_t) \, \delta^{\prime}_t \, \textnormal{d}t$. Therefore, for all $f \in \mathcal{Q}$,
\begin{equation*}
\begin{split}
\int \left[u_1(f_t)\tilde T^{\prime}_t + u_\lambda\(\mathtt{w}-f_t\)\right] \textnormal{d}t
\leq \int \left[u_1(f^{*}_t)\delta^{\prime}_t + u_\lambda\(\mathtt{w}-f^{*}_t\)\right] \textnormal{d}t 
= \int \left[u_1(f^{*}_t)\tilde T^{\prime}_t  + u_\lambda(\mathtt{w}-f^{*}_t)\right] \textnormal{d}t.
\end{split}
\end{equation*}

\noindent Thus,  $f^{*}$ solves Problem \eqref{eq:initial-problem}, and so
${X}^{*} = f^{*}(\mathtt{U})= m^{-1}_\lambda(\delta^\prime(\mathtt{U}))$ solves Problem \eqref{eq:gen-prob2}. The uniqueness in distribution of ${X}^*$ follows from \cite[Lemma A.8]{BBG24}.\qed

\bigskip

The following lemma was employed  for the proof of Lemma \ref{LemBBGnew}.

\begin{lemma}[Lemma A.5 in \cite{BBG24}]
Let $\delta$ be the convex envelope of $\tilde T$ on $\left[0,1\right]$. Then for any $f \in \mathcal{Q}$, we have 
$\displaystyle\int u_1\(f_t\) \, \tilde T^{\prime}_t \, \textnormal{d}t \leq \int u_1\(f_t\) \, \delta^{\prime}_t \, \textnormal{d}t.$
\label{fctQuantSolLemmaIneqAltINNBelHom2}
\end{lemma}

\medskip
\noindent{\bf{Proof of Corollary \ref{corsec2}:}}  This is a direct consequence of Proposition \ref{PropPOMax}, Lemma \ref{PropRepAg}, Lemma \ref{LemBBGnew}, and the observation that if $T$ is convex, then $\tilde T$ is concave and hence $\delta(t)=t$, for all $t\in[0,1]$.\qed

\bigskip
\section{Proofs for Section \ref{Section:noisePO}}

\noindent{\bf{Proof of Proposition \ref{Prop:MG}:}} 
The proof of the first part of this proposition is similar to that of \cite[Lemma A.8]{Ghossoub2019b}. We provide it below for the sake of completeness. 

\medskip

Suppose that $\tilde T$ is S-shaped with inflection point $t_0 \in [0,1]$. First, note that $\tilde T\(0\) = \delta_{\tilde T}\(0\) = 0$ and $\tilde T\(1\) = \delta_{\tilde T}\(1\) = 1$. Moreover, $\tilde T$ is convex on $\left[0,t_{0}\right)$, concave on $\left[t_{0},1\right]$, and increasing on $\left[0,1\right]$. Hence, $\tilde T^{\prime}\(t\) \geq 0$ for all $t \in \left[0,1\right]$, $\tilde T^{\prime}$ is nondecreasing on $\left[0,t_{0}\right)$, $\tilde T^{\prime}$ attains its maximum $\tilde T^{\prime}\(t_{0}\)$ at $t_{0}$, and $\tilde T^{\prime}$ is nonincreasing on $\left[t_{0},1\right]$.  

\medskip

Moreover, since $\tilde T$ is convex on $\left[0,t_{0}\right]$ and concave on $\left[t_{0},1\right]$, there exists some $z^{*} \in \left[0,1\right]$, which is unique by strict monotonicity of $\tilde T$, such that $\delta_{\tilde T}\(t\) = \tilde T\(t\)$ for all $t \in \left[0,z^{*}\right]$ and $\delta_{\tilde T}\(t\) \neq \tilde T\(t\)$ (i.e., $\delta_{\tilde T}\(t\) < \tilde T\(t\)$) for all $t \in \(z^{*},1\)$. Consequently, $z^{*} < t_{0}$, since otherwise $\delta_{\tilde T}$ will not be convex on all of the interval $\left[0,1\right]$. Also, $\delta_{\tilde T}$ is affine on $\left[z^{*},1\right]$, of the form $at+b$. Since $\delta_{\tilde T}\(1\) = 1$ and $\delta_{\tilde T}\(z_{*}\) = \tilde T\(z_{*}\)$, it follows that $a = \frac{1-\tilde T\(z^{*}\)}{1-z^{*}}$ and $b=1-a = \frac{\tilde T\(z^{*}\) - z^{*}}{1-z^{*}}$. Thus, 
\begin{equation}
\delta_{\tilde T}\(t\) = \left\{
\begin{array}{l l}
\tilde T\(t\) & \quad \mbox{if $t \in \left[0,z^{*}\right]$;}\\
\( \frac{1-\tilde T\(z^{*}\)}{1-z^{*}}\) t + \(\frac{\tilde T\(z^{*}\)-z^{*}}{1-z^{*}}\) = \tilde T\(z^{*}\) + \(\frac{\tilde T\(z^{*}\)-1}{z^{*}-1}\) \(t-z^{*}\) & \quad \mbox{if $t \in \left[z^{*},1\right]$.}\\
\end{array} \right.
\label{eqDeltaN}
\end{equation}

\medskip

\noindent Let $z_{0} \in \left[0,1\)$ be defined by $$z_{0} := \inf \left\{z \in \left[0,1\): \tilde T\(t\) > \tilde T\(z\) + \(\frac{1-\tilde T\(z\)}{1-z}\)\(t-z\), \hbox{ for all } t \in (z,1)\right\}.$$ Since $\tilde T$ is strictly concave on $\left(t_{0},1\right)$, the line segment connecting any two points $(t_{1},\tilde T\(t_{1}\))$ and $(t_{2},\tilde T\(t_{2}\))$, for $t_{1}, t_{2} \in \left[t_{0},1\right]$ with $t_{1} < t_{2}$, lies below the graph of $\tilde T$ (the epigraph of  $-\tilde T$ on the interval $\left[t_{0},1\right]$ is a convex set). In particular, for $t_{1} = t_{0}$ and $t_{2} = 1$, we have that for all $t \in (t_{0},1)$, $$\tilde T\(t\) > \tilde T\(t_{0}\) + \(\frac{\tilde T\(1\)-\tilde T\(t_{0}\)}{1-t_{0}}\)\(t-t_{0}\) = \tilde T\(t_{0}\) + \(\frac{1-\tilde T\(t_{0}\)}{1-t_{0}}\)\(t-t_{0}\).$$ Therefore, $t_{0} \geq z_{0}$, by definition of $z_{0}$. 

\medskip

Now, let $z_{1} := \inf \left\{z \in [0,1): \tilde T^{\prime}(z) \geq \frac{1-\tilde T\(z\)}{1-z}\right\}.$ Define the subsets $\mathcal{A},\mathcal{B}\subset \left[0,1\right]$ by 
\begin{eqnarray*}
\mathcal{A} := \left\{z : \tilde T\(t\) > \tilde T\(z\) + \frac{1-\tilde T\(z\)}{1-z}\(t-z\), \forall t \in (z,1)\right\}, \quad \textnormal{ and }  \quad   
\mathcal{B} := \left\{z : \tilde T^{\prime}\(z\) \geq \frac{1-\tilde T\(z\)}{1-z}\right\}.
\end{eqnarray*}
 Since $T$ is continuously differentiable on [0,1], so is $\tilde T$, and so for every $z \in \mathcal{A}$,
$$\tilde T^{\prime}\(z\) = \underset{t \to z}\lim \, \frac{\tilde T\(t\) - \tilde T\(\tilde z\) }{t-z} = \underset{t \downarrow z}\lim \, \frac{\tilde T\(t\) - \tilde T\(\tilde z\) }{t-z}  \geq \frac{1-\tilde T\(z\)}{1-z},$$
implying that $\mathcal{A} \subseteq \mathcal{B}$. If $\mathcal{A} \subsetneq \mathcal{B}$, there exists $\tilde z$ such that $\tilde T\(t\) > \tilde T(\tilde z) + (\frac{1-\tilde T\(\tilde z\)}{1-\tilde z})\(t-\tilde z\)$, for all $t \in (\tilde z,1)$, and $\tilde T^{\prime}\(\tilde z\) < \frac{1-\tilde T\(\tilde z\)}{1-\tilde z}$, implying that $\tilde T^{\prime}\(\tilde z\) \geq \frac{1 - \tilde T\(\tilde z\)}{1 - \tilde z}$ and $\tilde T^{\prime}\(\tilde z\) < \frac{1 - \tilde T\(\tilde z\)}{1 - \tilde z}$, a contradiction. Therefore, $\mathcal{A} = \mathcal{B}$, and so $z_{0} = z_{1}$. 

\medskip

We next show that $z^{*} = z_{0} = z_{1}$. Suppose, by way of contradiction, that $z^{*} \neq z_{0}$. If $z^{*} < z_{0}$, then there exists some $m \in (z^{*},1)$ such that $\tilde T\(m\) \leq \tilde T\(z^{*}\) + (\frac{1 - \tilde T(z^{*})}{1 - z^{*}}) \(m-z^{*}\)  = \delta_{\tilde T}^{*}(m)$, contradicting the fact that $\delta_{\tilde T}\(t\) < \tilde T\(t\)$, for all $t \in \(z^{*},1\)$. Now, suppose that $z^{*} > z_{0} = z_{1}$, and fix some $\bar z \in (z_0,z^*)$. Then $\tilde T(\bar z) = \delta_{\tilde T}(\bar z)$, $\tilde T^{\prime}\(\bar z\) \geq \frac{1 - \tilde T\(\bar z\)}{1-\bar z}$, and $\tilde T\(t\) > \tilde T\(\bar z\) + (\frac{1-\tilde T\(\bar z\)}{1-\bar z})(t-\bar z)$, for all  $t \in (\bar z,1)$. In particular, 
$$
\delta_{\tilde T}(z^*) = \tilde T\(z^{*}\) 
> \tilde T\(\bar z\) + \(\frac{1-\tilde T\(\bar z\)}{1-\bar z}\)\(z^{*}-\bar z\)
= \delta_{\tilde T}(\bar z) + \(\frac{1-\delta_{\tilde T}(\bar z)}{1-\bar z}\)\(z^{*}-\bar z\).
$$

\noindent However, since $\delta_{\tilde T}$ is convex on $[0,1]$, it lies below the line segment connecting the points $\(\bar z, \delta_{\tilde T}(\bar z)\)$ and $\(1, \delta_{\tilde T}(1)\) = \(1, 1\)$. Therefore, 
$$\delta_{\tilde T}(z^*) 
\leq
\delta_{\tilde T}(\bar z) + \(\frac{1-\delta_{\tilde T}(\bar z)}{1-\bar z}\)\(z^{*}-\bar z\),$$
a contradiction. Hence, $z^* = z_0 = z_1$, and so $\delta_{\tilde T}$ is given by:
\begin{equation}
\delta_{\tilde T}\(t\) = \left\{
\begin{array}{l l}
\tilde T\(t\) & \quad \mbox{if $t \in \left[0,z_{0}\right]$;}\\
\tilde T\(z_{0}\) + \(\frac{1 - \tilde T\(z_{0}\)}{1 - z_{0}}\) \(t-z_{0}\) & \quad \mbox{if $t \in \left[z_{0},1\right]$.}\\
\end{array} \right.
\label{eqDeltaz0bN}
\end{equation}

\medskip

Finally, since $\delta_{\tilde T}$ is convex on the interval $\left[0,1\right]$, the line segment connecting the two points $\(0,\delta_{\tilde T}\(0\)\) = \(0,0\)$ and $\(0,\delta_{\tilde T}\(1\)\) = \(1,1\)$ lies above the graph of $\delta_{\tilde T}$. However, this line segment is the graph of the identity function on $\left[0,1\right]$. Consequently, $\delta_{\tilde T}\(t\) \leq t$, for all $t \in \left[0,1\right]$. In particular, since $\tilde T\(z_{0}\) = \delta_{\tilde T}\(z_{0}\)$ by eq.\ \eqref{eqDeltaz0bN}, we have $ \tilde T\(z_{0}\) \leq z_{0}$.

\medskip

The proof of the second part of Proposition \ref{Prop:MG} follows by a symmetric argument.\qed
\medskip

\bigskip

\noindent{\textbf{Proof of Proposition  \ref{prop:pdf}:}}     
Let $\mathtt{U}  \sim Uni(0,1)$. By Theorem \ref{ThMain}, we have ${X}_1 =m_\lambda^{-1}( \delta^\prime(\mathtt{U}))$, and by Proposition \ref{Prop:MG}, since $T$ is inverse S-shaped, the convex envelope $\delta$ of $\tilde T$ satisfies
$$
\delta^\prime(U)
=
\frac{\tilde T(p^*)}{p^*}\,\mathbf 1_{\{U <  p^*\}}
+
\tilde T^\prime(p)\,\mathbf 1_{\{U \geq p^*\}}.
$$
Consequently,
$$
X_1 = x_0 \, \mathbf 1_{\{U < p^*\}}
+
m_\lambda^{-1}\left(\tilde T^\prime(U)\right)\,\mathbf 1_{\{U \geq p^*\}},
$$
where
$$
x_0 := m_\lambda^{-1}\,\left(\frac{\tilde T(p^*)}{p^*}\right).
$$
Hence, $X_1$ has an atom of mass $p^* = \p[U < p^*]$ at $x_0$. On the event $\{U \geq p^*\}$, define the transformation
$$I(p):=m_\lambda^{-1}\left(\tilde T^\prime(p)\right),$$
for $p\in [p^*,1]$, so that
$$
X_1 = x_0 \, \mathbf 1_{\{U < p^*\}}
+
I\left(U\right)\,\mathbf 1_{\{U \geq p^*\}}.
$$

Since $\tilde T^\prime$ and $m_\lambda$ are strictly increasing, $I$ is strictly increasing and continuously differentiable on $[p^*,1]$. Moreover, for $x\in I([p^*,1])$, 
$$
f_1(x) := f_{X_1}(x) = f_{I(U)}(x) = \frac{d}{dx} \p[I(U) \leq x] =
\left|(I^{-1})^\prime(x)\right|.
$$
Since $I^{-1}(x) = (\tilde T')^{-1}\left(m_\lambda(x)\right)$, we have
$$(I^{-1})^\prime(x) = \frac{d}{dx}(\tilde T')^{-1}\left(m_\lambda(x)\right).
$$

\noindent Now, strict convexity of $\tilde T$ on $[p^*,1]$ implies $\tilde T''>0$ on $[p^*,1]$. Monotonicity of $m_\lambda$ gives $m_\lambda'>0$. Therefore,
\begin{align*}
f_1(x) 
&= \left| \frac{d}{dx}(\tilde T')^{-1}\left(m_\lambda(x)\right)\right|
= \left|\frac{m_\lambda^\prime(x)}
{\tilde T''\left((\tilde T')^{-1}(m_\lambda(x))\right)}\right|
= \frac{m_\lambda^\prime(x)}
{\tilde T''\left((\tilde T')^{-1}(m_\lambda(x))\right)}.
\end{align*}

\noindent Since $\tilde T'(t)=T'(1-t)$, we have $(\tilde T')^{-1}(y)=1-(T')^{-1}(y)$. Therefore,
$$\tilde T''\left((\tilde T')^{-1}(m_\lambda(x))\right)
=
- T''\left(1-(\tilde T')^{-1}(m_\lambda(x))\right)
=
-T''\left((T')^{-1}(m_\lambda(x))\right),
$$
and hence 
$$f_1(x)=
\frac{- m_\lambda'(x)}
{T''\left((T')^{-1}(m_\lambda(x))\right)},$$

\noindent for $x\in I([p^*,1])$. \qed

\bigskip

\noindent{\bf{Proof of Corollary \ref{cor:palpha}:}}
\

(1)    By Proposition \ref{Prop:MG}.1, for $S$-shaped Prelec distortion functions ($\alpha > 1$), we have
for the first deviation point, $p^*(\alpha)$, of the convex envelope with respect to  $ \tilde T$ is given by
$$p^*(\alpha) = \inf \left\{ p\in [0, 1) : \tilde T'(p) \geq \frac{\tilde T(p) - 1}{p - 1} \right\},$$ 
where $T\(p\) = \exp\(-\(-\ln\(p\)\)^{\alpha}\), \ \forall p \in [0,1],$ 
with $\alpha> 1$. Then
\begin{align*}
&T^\prime(p) = \frac{-\alpha}{p \, \ln(p)} \, \exp\(-\(-\ln\(p\)\)^{\alpha}\) \, \(-\ln(p)\)^\alpha \geq 0,\\
&\tilde T(p) = 1 - \exp\(-\(-\ln\(1-p\)\)^{\alpha}\),\\
&\tilde T^\prime(p) =  \frac{-\alpha}{(1-p) \, \ln(1-p)} \, \exp\(-\(-\ln\(1-p\)\)^{\alpha}\) \, \(-\ln(1-p)\)^\alpha \\
&\quad\quad =  \frac{\alpha}{1-p} \, \exp\(-\(-\ln\(1-p\)\)^{\alpha}\) \, \(-\ln(1-p)\)^{\alpha-1} \geq 0.
\end{align*}

\noindent Hence,
\begin{align*}
p^*(\alpha) 
&= \inf \left\{ p\in [0, 1) : \frac{\alpha}{1-p} \, \exp\(-\(-\ln\(1-p\)\)^{\alpha}\) \, \(-\ln(1-p)\)^{\alpha-1} \geq \frac{\exp\(-\(-\ln\(1-p\)\)^{\alpha}\)}{1 - p} \right\}.
\end{align*}

\bigskip

\noindent For the S-shaped Prelec function that we consider here, the functions $\tilde T'(p)$ and $\frac{1 - \tilde T(p)}{1 - p}$ cross once, and the point at which they cross is precisely $p^*(\alpha)$, so that
$$\tilde T^\prime\(p^*(\alpha)\) = \frac{1 - \tilde T\(p^*(\alpha)\)}{1 - p^*(\alpha)},$$
that is
$$\frac{\alpha}{1-p^*(\alpha)} \, \exp\(-\(-\ln\(1-p^*(\alpha)\)\)^{\alpha}\) \, \(-\ln(1-p^*(\alpha))\)^{\alpha-1}
=
\frac{\exp\(-\(-\ln\(1-p^*(\alpha)\)\)^{\alpha}\)}{1 - p^*(\alpha)},$$

\bigskip

\noindent which yields the desired closed form 
$p^*(\alpha)
= 1 - 
\exp(-\(\frac{1}{\alpha}\)^{\frac{1}{\alpha-1}}).$

\bigskip
   (2) 
By Proposition \ref{Prop:MG}(2)  for an inverse S-shaped Prelec distortion function, with $\alpha < 1$, we have 
\begin{equation*}
\begin{split}
p^*(\alpha) 
&= \sup \left\{p \in \left[0,1\): \tilde T^{\prime}\(p\) \leq \frac{\tilde T\(p\)}{p}\right\}\\
&= \sup \left\{p \in \left[0,1\): \frac{\alpha}{1-p} \, \exp\(-\(-\ln\(1-p\)\)^{\alpha}\) \, \(-\ln(1-p)\)^{\alpha-1} \leq \frac{1 - \exp\(-\(-\ln\(1-p\)\)^{\alpha}\)}{p}\right\}.\\
\end{split}
\end{equation*}

\noindent The functions $\tilde T'(p)$ and $\frac{\tilde T(p)}{p}$ cross once, and the point at which they cross is precisely $p^*(\alpha)$, so that
$\tilde T^\prime\(p^*(\alpha)\) = \frac{\tilde T\(p^*(\alpha)\)}{p^*(\alpha)},$
that is
\begin{eqnarray}\label{palpha}&&
\frac{\alpha}{1-p^*(\alpha)} \, \exp\(-\(-\ln\(1-p^*(\alpha)\)\)^{\alpha}\) \, \(-\ln(1-p^*(\alpha))\)^{\alpha-1} 
\\&=&
\frac{1 - \exp\(-\(-\ln\(1-p^*(\alpha)\)\)^{\alpha}\)}{p^*(\alpha)}.\nonumber
\end{eqnarray}

\medskip

Now  define  $
x = -\ln(1 - p^*(\alpha))$. 
Then we can rewrite $ 1 - p^*(\alpha)$ as 
$
1 - p^*(\alpha) = e^{-x}$. 
Rewriting \eqref{palpha}  in terms of $x$ yields
\[
\frac{\alpha}{e^{-x}} e^{  - x^\alpha } x^{\alpha-1}
=
\frac{1 - e^{-x^\alpha}}{1 - e^{-x}}  \Longleftrightarrow  \alpha e^{-x^\alpha} x^{\alpha-1} = \frac{1 - e^{-x^\alpha}}{e^x(1 - e^{-x})}.
\]
Multiplying both sides by $ e^x (1 - e^{-x})$ yields $
\alpha x^{\alpha-1} e^{-x^\alpha} e^x (1 - e^{-x}) = 1 - e^{-x^\alpha}
$. 
Rearrange to make $ e^{-x^\alpha}$ dividing both sides by $ \alpha x^{\alpha-1} e^x (1 - e^{-x}) + 1 $ gives
\[
e^{-x^\alpha} = \frac{1}{\alpha x^{\alpha-1} e^x (1 - e^{-x}) + 1},
\]
and so: 
\[
x^\alpha = -\ln \left( \frac{1}{\alpha x^{\alpha-1} e^x (1 - e^{-x}) + 1} \right) = \ln \left( \alpha x^{\alpha-1} e^x (1 - e^{-x}) + 1 \right).\]
This gives an implicit definition for $p^*(\alpha) $, via $x = -\ln(1 - p^*(\alpha))$, for $
p^*(\alpha) = 1 - e^{-x}$. \qed

\bigskip

\noindent{\bf{Proof of Corollary  \ref{Cor_explict}:}} 
Note that $X_1
= m^{-1}_\lambda\(\delta^\prime(\mathtt{U})\)$, by Theorem \ref{ThMain}. In that case, it holds that $u_\lambda(x)=-\frac{1}{\overline \beta}e^{-\overline\beta x}\sum_{i=2}^n\lambda_i$.
We then have  $m_\lambda(x) =e^{\(\beta_1+\overline\beta\) x}\sum_{i=2}^n\lambda_i$ for $x\in\mathbb{R}$, and so 
$$
m^{-1}_\lambda(y) = \frac{1}{\beta_1+\overline \beta} \, \left(\ln(y)-\ln(\sum_{i=2}^n\lambda_i)\right),
$$
for $y>0$. In general, $m_\lambda^{-1}$ is  increasing  with $\lim_{y\rightarrow 0}m_\lambda^{-1}(y)=-\infty$ and $\lim_{y\rightarrow \infty}m_\lambda^{-1}(y)=\infty$.\qed

\bigskip

\section{Proofs for Section \ref{Sec:nud}}

\noindent{\bf{Proof of Lemma \ref{L:TM}:}}  
By the definition of $T_M$, we have 
 \begin{align*}
\tilde T_M\(p\)
= 1 - T_M\(1-p\)
&= 1 - \left((1-f(M))T(1-p)+f(M) (1-p)\right)\\
&= 1 - T(1-p) + f(M)T(1-p) - f(M) (1-p)\\
&= \tilde T(p) - f(M) \left[\tilde T(p) - p\right],
\end{align*}

and thus $
\tilde T_M^\prime\(p\)
= \tilde T^\prime(p) - f(M) \left[\tilde T^\prime(p) - 1\right]
= \tilde T^\prime(p) \left[1 - f(M)\right] + f(M)$. Therefore,
\begin{align*}
U_1(X, T_M)
&=\displaystyle\int_0^1 u_1\(F_{X}^{-1}\(p\)\) \tilde T_M^\prime\(p\) \, \textnormal{d}p \\
&=\displaystyle\int_0^1 u_1\(F_{X}^{-1}\(p\)\)\left\{ \tilde T^\prime(p) \left[1 - f(M)\right] + f(M) \right\} \, \textnormal{d}p \\
&= \left[1 - f(M)\right] \displaystyle\int_0^1 u_1\(F_{X}^{-1}\(p\)\) \tilde T^\prime(p) \, \textnormal{d}p 
+
f(M) \, \displaystyle\int_0^1 u_1\(F_{X}^{-1}\(p\)\)   \, \textnormal{d}p \\
&= \left[1 - f(M)\right] \, U_1(X, T) + f(M) \, \int u_1(X) \, \textnormal{d}\p.
\end{align*}
\qed

\bigskip

{\noindent{\bf{Proof of Lemma \ref{L:delta}:}}  
Note that $T:[0,1]\to[0,1]$ is increasing, with $T(0)=0$ and $T(1)=1$. Define $\tilde T: [0,1] \to [0,1]$ by $\tilde T(t):=1-T(1-t)$. For $f(M)\in[0,1]$, recall that $T_M: [0,1] \to [0,1]$ is defined by $T_M (t) := \left(1-f(M)\right) \, T(t) + f(M)\, t$, and $\tilde T_M: [0,1] \to [0,1]$ is given  by
$$
\tilde T_M(t) 
:= 1 - T_M(1-t) 
=\left(1-f(M)\right) \, \tilde T(t) + f(M)\,t.
$$

Recall that the convex envelope $\delta$ of $\tilde T$ on $[0,1]$ is defined as the largest convex function that satisfies $\delta(t) \leq \tilde T(t)$, for all  $t \in [0,1]$: 
$$
\delta(t) 
= \sup \, \left\{h(t) : h \hbox{ is convex on } [0,1], \  h(t) \leq \tilde T(t), \ \forall \, t \in [0,1]\right\}.
$$

\noindent Similarly, $\delta_M$ is the convex envelope of $\tilde T_M$. Define the function $g$ on $[0,1]$ by
$$
g(t) := \left(1-f(M)\right)\, \delta(t) + f(M) \, t.
$$

\noindent We show that $g = \delta_M$ by showing that $g$ satisfies the definition of the convex envelope of $\tilde T_M$.

\medskip

 \textbf{\bf $g$ is convex:} Since $f(M) \in [0,1]$, $g$ is a convex combination of the convex function $\delta$, and the affine (hence convex) function $t \mapsto t$. A convex combination of convex functions is convex, and therefore $g$ is convex on $[0,1]$.

\medskip

 {\bf $g \leq \tilde T_M$ on $[0,1]$:} For any $t \in [0,1]$, we have $\delta(t) \leq \tilde T(t)$, by the definition of the convex envelope. Since $f(M) \in [0,1]$, it follows that
$$
g(t) 
= \left(1-f(M)\right)\,\delta(t) + f(M)\,t 
\leq \left(1-f(M)\right)\,\tilde T(t) + f(M)\,t 
= \tilde T_M(t).
$$

\medskip
 {\bf $g$ is maximal:} For any convex function  $h:[0,1] \to \mathbb{R}$ with $h \leq \tilde T_M$ on $[0,1]$, we have $h \leq g$.

\medskip
(1) {\bf If $f(M) < 1$:} Define the function $h_0$ by  
$$
h_0(t) := \frac{h(t) - f(M)\,t}{1-f(M)}.
$$

\medskip

\noindent We claim that $h_0$ is convex and $h_0 \leq \tilde T$ on $[0,1]$.

\medskip

\begin{enumerate}
\item[(a)] {\bf $h_0$ is convex:}
For any $t_1, t_2 \in [0,1]$ and $\lambda \in [0,1]$, the convexity of $h$ gives
$$
h(\lambda t_1 + (1-\lambda)t_2) \leq \lambda \, h(t_1) + (1-\lambda) \, h(t_2).
$$

\noindent Let $t_\lambda := \lambda t_1 + (1-\lambda)t_2$. Then
\begin{align*}
h_0(t_\lambda) 
= \frac{h(t_\lambda) - f(M)\,t_\lambda}{1-f(M)} 
&\leq \frac{\lambda \, h(t_1) + (1-\lambda)\,h(t_2) - f(M)(\lambda t_1 + (1-\lambda)\,t_2)}{1-f(M)} \\
&= \frac{\lambda \, [h(t_1) - f(M) \, t_1] + (1-\lambda)\,[h(t_2) - f(M)\, t_2]}{1-f(M)} \\
&= \lambda \, h_0(t_1) + (1-\lambda) \, h_0(t_2),
\end{align*}

\noindent and hence $h_0$ is convex.

\bigskip

\item[(b)] {\bf $h_0 \leq \tilde T$ on $[0,1]$:}
Since $h \leq \tilde T_M$ by hypothesis, we have for all $t \in [0,1]$,
$$
h(t) 
\leq \tilde T_M(t) = \left(1-f(M)\right)\,\tilde T(t) + f(M)\, t,
$$

\noindent and hence
$
h(t) - f(M)\,t \leq \left(1-f(M)\right)\,\tilde T(t).
$ Since $1-f(M) > 0$, it follows that 
$$
h_0(t) = \frac{h(t) - f(M)t}{1-f(M)} \leq \tilde T(t).
$$
\end{enumerate}

\medskip

By (a) and (b), the maximality of the convex envelope $\delta$ implies that $h_0 \leq \delta$ on $[0,1]$. Therefore,
$$
h(t) 
= \left(1-f(M)\right) \, h_0(t) + f(M) \,t 
\leq \left(1-f(M)\right) \, \delta(t) + f(M) \, t 
= g(t).
$$
\medskip

(2) {\bf If $f(M) = 1$:}
In this case, $\tilde T_M(t) = t$ and $g(t) = t$. Since $t \mapsto t$ is convex and equals $\tilde T_M$, we have $\delta_M(t) = t = g(t)$.

\medskip

\noindent Consequently, $g$ is the largest convex function bounded above by $\tilde T_M$, which means that $g = \delta_M$.

\medskip

Finally, since $\widetilde T$ is bounded on $[0,1]$ and $\delta$ is its convex envelope,
$\delta$ is finite and convex on $[0,1]$, with $\delta(0)=\widetilde T(0)=0$ and $\delta(1)=\widetilde T(1)=1$. Any finite convex function on a compact interval is locally Lipschitz continuous on the interior of its domain, and hence absolutely continuous on every compact subinterval of $(0,1)$. In particular, $\delta$ is absolutely continuous on $(0,1)$ and therefore differentiable a.e.\ on $(0,1)$. Now, since $\delta_M(t)=(1-f(M))\,\delta(t)+f(M)\,t$, $\delta_M$ is also finite and convex on $[0,1]$, and it is absolutely continuous on $(0,1)$. Linearity of differentiation under scaling and addition then yields $\delta_M'(t)=(1-f(M))\,\delta'(t)+f(M)$, for a.e.\ $t\in(0,1)$. \hfill $\square$
}

\bigskip

\noindent\textbf{Proof of Proposition} \ref{ThMain2}: \ 

(1) For a given effort level $M$, applying Theorem \ref{ThMain} to the two-agent economy with aggregate endowment $\mathtt{w}-M$, the allocation $\mathbf{X}^M$ solves Problem \eqref{eqgen-prob2} if and only if
\begin{align*}
X_1^M =m^{-1}_{\lambda_2}\( \delta_M^\prime(\mathtt{U})\) \ \ \hbox{and} \ \  X_2^M= \mathtt{w}-M-{X}_1^M,
\end{align*}

\noindent where $m_{\lambda_2}(x):=\lambda_2 \, \frac{u_2^\prime\(\mathtt{w}-M-x\)}{u_1^\prime\(x\)}$. By Lemma \ref{L:delta}, we have
$$\delta_M'(t) 
= (1-f(M)) \, \delta'(t)+f(M)
=\delta'(t) - f(M) \left[\delta'(t)-1\right],$$ 

\noindent for a.e.\ $t\in(0,1)$. Consequently,
$$X_1^M = m_{\lambda_2}^{-1}\left(\delta'(U) - f(M)\left[\delta'(U) - 1\right]\right).$$

\medskip

(2) Note that the function $m_{\lambda_2}(x)=\lambda_2\,\frac{u_2^\prime(\mathtt{w}-M-x)}{u_1^\prime(x)}$ depends explicitly on $M$. Consequently, its inverse $m_{\lambda_2}^{-1}$ also depends on $M$. To emphasize this dependence, we write $m_{\lambda_2}\left(X_1^M;M\right)=\phi(M)$, where $\phi(M):=\delta^\prime(U)-f(M)\left(\delta^\prime(U)-1\right)$. The chain rule gives
$$
\phi^\prime(M)
=
\frac{\partial}{\partial X_1^M} \, m_{\lambda_2}\left(X_1^M;M\right) \, \frac{\partial X_1^M}{\partial M}
+
\frac{\partial}{\partial M} \, m_{\lambda_2}\left(X_1^M;M\right).
$$
Hence,
$$
\frac{\partial X_1^M}{\partial M}
=\frac{\phi^\prime(M) - \frac{\partial}{\partial M} \, m_{\lambda_2}\left(X_1^M;M\right)}
{\frac{\partial}{\partial X_1^M} \, m_{\lambda_2}\left(X_1^M;M\right)}.
$$
Since
$$
\phi^\prime(M)=-f^\prime(M)\left(\delta^\prime(U)-1\right)
\ \ \hbox{and} \ \ 
\frac{\partial}{\partial M} \, m_{\lambda_2}\left(X_1^M;M\right)
=\lambda_2\,\frac{-u_2^{\prime\prime}(\mathtt{w}-M-x)}{u_1^\prime(x)},
$$
we obtain
$$
\frac{\partial X_1^M}{\partial M}
=
\frac{-f^\prime(M)\left(\delta^\prime(U)-1\right)
+\lambda_2\,\dfrac{u_2^{\prime\prime}(\mathtt{w}-M-X_1^M)}{u_1^\prime(X_1^M)}}
{\frac{\partial}{\partial X_1^M} \, m_{\lambda_2}\left(X_1^M;M\right)}.
$$

\noindent Writing $X_1^M=m_{\lambda_2}^{-1}(\phi(M))$ and $\Lambda(x):=\frac{d}{dx}m_{\lambda_2}\big(m_{\lambda_2}^{-1}(x)\big)$, we obtain
\begin{align*}
\frac{\partial X_1^M}{\partial M}
&=
\frac{-f^\prime(M)\left(\delta^\prime(U)-1\right)
+\lambda_2\,\dfrac{u_2^{\prime\prime}\left(\mathtt{w}-M-m_{\lambda_2}^{-1}(\phi(M))\right)}{u_1^\prime\left(m_{\lambda_2}^{-1}(\phi(M))\right)}}
{\Lambda(\phi(M))} \\
&=
\frac{-f^\prime(M)\left(\delta^\prime(U)-1\right)
+\lambda_2\,\dfrac{u_2^{\prime\prime}\left(\mathtt{w}-M-m_{\lambda_2}^{-1}\left(\delta^\prime(U)-f(M)\left(\delta^\prime(U)-1\right)\right)\right)}{u_1^\prime\left(m_{\lambda_2}^{-1}\left(\delta^\prime(U)-f(M)\left(\delta^\prime(U)-1\right)\right)\right)}}
{\Lambda\left(\delta^\prime(U)-f(M)\left(\delta^\prime(U)-1\right)\right)}.
\end{align*}

\noindent Finally, since
$$
\frac{d }{dx} m_{\lambda_2}(x)
= \lambda_2 \left[
\frac{-u_2''(\mathtt{w} - M - x)}{u_1'(x)}
- \frac{u_2'(\mathtt{w} - M - x)\,u_1''(x)}{(u_1'(x))^2}
\right],
$$
it follows that
\begin{align*}
\Lambda(x)
&=\lambda_2\left[
\frac{-u_2''\,\big(\mathtt{w}-M-m_{\lambda_2}^{-1}(x)\big)}{u_1'\,\big(m_{\lambda_2}^{-1}(x)\big)}
-\frac{u_2'\,\big(\mathtt{w}-M-m_{\lambda_2}^{-1}(x)\big)\,u_1''\,\big(m_{\lambda_2}^{-1}(x)\big)}
{\Big(u_1'\,\big(m_{\lambda_2}^{-1}(x)\big)\Big)^2}
\right]\\
&=\frac{- \lambda_2}{u_1'\left(m_{\lambda_2}^{-1}(x)\right)}
\left[
u_2''\left(\mathtt{w} - M - m_{\lambda_2}^{-1}(x)\right)
+
u_2'\left(\mathtt{w} - M - m_{\lambda_2}^{-1}(x)\right)
\frac{u_1''\left(m_{\lambda_2}^{-1}(x)\right)}
{u_1'\left(m_{\lambda_2}^{-1}(x)\right)}
\right].
\end{align*}
\qed

\bigskip
\noindent{\bf{Proof of  Theorem \ref{welfareTHM}:}} 
 The welfare maximization with a control on $M$ is formally given by: 
\begin{align}\label{nudge2}
\max_{M\in [0,\mathtt{w}]} \, \sup_{X_1+X_2=\mathtt{w}-M} \left[ U_1(X_1, T_M)+U_2(X_2)\right]
=\max_{M\in [0,\mathtt{w}]} \, \sup_{X \in \cX} \, \left[ U_1(X, T_M)+U_2(\mathtt{w}-M-X)\right].
\end{align}

\medskip

\noindent For a given  effort level $M$, the optimal risk sharing rule  $(X_M, \mathtt{w}-M-X_M )$ satisfies    
\begin{align*}
 W(X_M,M)
&=\max_{X} \left\{U_1(X, T_M) + U_2(\mathtt{w}-M-X)\right\} \\     
&=U_1(X_M, T_M) + U_2(\mathtt{w}-M - X_M)\\
&=\left(1 - f(M)\right) \, U_1(X_M, T) +  \,E^\p\left[f(M)  u_1(X_M) +u_2(\mathtt{w}-M-X_M)\right].
\end{align*}

\medskip

\noindent To find the optimal monetary value $M^*$, we need to find the maximum of  $M \mapsto W (X_M, M)$, by taking first order conditions $\frac{\partial}{\partial M} \, W(X^M,M)=0$. To do so, we first rewrite $W(X^M,M)$ and then compute $\frac{\partial}{\partial M} \, W(X^M,M)$. Recall that $T_M(t)= (1-f(M)) \, T(t) + f(M) \, t,$ 
so that
\begin{align*}
\tilde T_M(t) = 1 - (1-f(M)) \, T(1-t) - f(M) \, (1-t)
\quad \textnormal{and} \quad 
\tilde T_M^\prime(t) = f(M) + (1-f(M)) \, \tilde T^\prime(t).
\end{align*}

\noindent Therefore,
\begin{align*}
W(X^M,M)
&=\left(1 - f(M)\right) \, U_1(X_M, T) +  \,E^\p[f(M)  u_1(X_M) +u_2(\mathtt{w}-M-X_M)]\\
&=\left(1 - f(M)\right) \, \int u_1(X_M) \, d T \circ \p 
\,+\,  
\int[f(M) \, u_1(X_M) +u_2(\mathtt{w}-M-X_M)] \,  d\p \\
&=\left(1 - f(M)\right) \, \int_0^1 u_1(F_{X_M}^{-1}(t)) \, T^\prime(1-t) \, dt 
\,+\,  
f(M) \, 
\int_0^1 u_1(F_{X_M}^{-1}(t)) \, dt 
+ \int_0^1 u_2(F^{-1}_{\mathtt{w}-M-X_M}(t)) \, dt\\
&=\int_0^1 \left[(1 - f(M)) \, u_1(x_M(t)) \, \tilde T^\prime(t) + f(M) \, u_1(x_M(t))
 + \ u_2(\mathtt{w}-M - x_M(1-t)) \right]\, dt.\\
&=\int_0^1 u_1(x_M(t)) \, \left[(1 - f(M)) \,  \tilde T^\prime(t) + f(M)\right] \, dt
 + \int_0^1  u_2(\mathtt{w}-M - x_M(1-t)) \, dt\\
&=\int_0^1 u_1\left(x_M(t)\right) \, \tilde T^\prime_M(t) \, dt
 + \int_0^1  u_2(\mathtt{w}-M - x_M(1-t)) \, dt,
\end{align*}

\noindent where we used 
$F^{-1}_{X_M}(t) = x_M(t)$, $F^{-1}_{\mathtt{w}-M-X_M}(t) = \mathtt{w} - M - F^{-1}_{X_M}(1-t)$, 
and $T^\prime(1-t) = \tilde T^\prime(t)$. Hence, letting $x^\prime_M (t) : =\frac{\partial }{\partial M} \, x_M(t)$, we have

\begin{align*}
\frac{\partial}{\partial M} \, W(X^M,M)
&=\int_0^1 {\frac{\partial}{\partial M} \, \left[u_1(x_M(t)) \, \tilde T^\prime_M(t) \right]} \, dt
+ \int_0^1   \frac{\partial}{\partial M}    u_2(\mathtt{w}-M - x_M(1-t))  dt\\
&=\int_0^1 \left\{u_1'(x_M(t)) x_M'(t) \tilde{T}_M'(t) + u_1(x_M(t)) f^\prime(M) \left[ 1 - \tilde{T}^\prime(t) \right]\right\}
 dt \\
&\qquad\qquad\qquad
- \int_0^1 u_2^\prime(\mathtt{w}-M - x_M(1-t)) \, \left[1 + x^\prime_M (1-t)\right] \, dt.
\end{align*}

\medskip

\noindent Thus, the FOC is given by:
$$\int_0^1 \frac{\partial}{\partial M} \, \left[u_1(x_M(t)) \, \tilde T^\prime_M(t) \right] \, dt =  -\int_0^1  \frac{\partial}{\partial M} \, u_2(\mathtt{w}-M - x_M(1-t))  \, dt.$$

\medskip

\noindent Of course, a sufficient condition for the FOC to hold is 
\begin{eqnarray}
\label{FOC}
\frac{\partial}{\partial M} \, \left[u_1(x_M(t)) \, \tilde T^\prime_M(t) \right]
= -  \frac{\partial}{\partial M} \, u_2(\mathtt{w}-M - x_M(1-t)) , \ \ \, \forall \, t \in (0,1).
\end{eqnarray}

\noindent Note that \eqref{FOC} can be written as
$$
u_1'(x_M(t)) x_M'(t) \tilde{T}_M'(t) + u_1(x_M(t)) f^\prime(M) \left[ 1 - \tilde{T}^\prime(t) \right] = u_2'(\mathtt{w} - M - x_M(1-t)) \left[ 1 + x_M'(1-t) \right].\qed
$$

\end{appendices}

\vspace{0.5cm}
\bibliographystyle{ecta}
\bibliography{Bib}
\vspace{0.6cm}

\end{document}